\documentclass[preprint,12pt]{elsarticle}

 \usepackage{amsthm}

\usepackage{amsmath,amssymb,soul}

\usepackage{enumitem}
\usepackage{bussproofs,graphicx,float,tikz}
\usepackage{pdfpages}
\usepackage{algorithm,graphicx}
\usepackage{mdframed}
\usepackage{mathpartir}
\usepackage{algpseudocode,tcolorbox}
\usepackage{blindtext,multicol,stmaryrd}
\usepackage{wrapfig}
\usepackage{orcidlink}
\usepackage{lipsum}

\usepackage{siunitx}

\usepackage[english]{babel}

\usepackage{tensor}
\usepackage{url}
\usepackage{lscape}
\usepackage{esvect,xcolor}
\usepackage{xy}
\xyoption{all}
\xyoption{knot}
\xyoption{arc}
\xyoption{import}
\xyoption{poly}

\usepackage{tikz}
\usepackage{tkz-euclide}

\usepackage{lipsum}

\usepackage{array}

\usepackage{stmaryrd}
\SetSymbolFont{stmry}{bold}{U}{stmry}{m}{n}

\usepackage{graphics} %needed for the NEW quantifier

\newcommand{\uP}{\mathcal{P}}

\newcommand{\aleq}{\approx_{\alpha}}
\newcommand{\caleq}{\approx_{\theory{C}}}
\newcommand{\ealeq}{\approx_{\theory{E}}}
\newcommand{\taleq}{\approx_{\theory{T}}}

\newcommand{\enarrow}[1]{\rightsquigarrow_{\theory{#1}}}

\newcommand\NEW[0]{\reflectbox{\ensuremath{\mathsf{N}}}}
\newcommand\nw[1]{#1^{\scalebox{.4}{\NEW}}}

\definecolor{bostonuniversityred}{rgb}{0.8, 0.0, 0.0}

\newcommand{\atomSet}{\mathbb{A}}

\newcommand{\pair}[1]{\langle #1\rangle}

\newcommand{\act}{\cdot}

\newcommand{\dom}[1]{\mathtt{dom}(#1)}

\newcommand{\vran}[1]{\mathtt{VRan}(#1)}

\newcommand{\nfpair}[1]{\langle #1\rangle_{nf}}
\newcommand{\pAction}[2]{#1\cdot #2}
\newcommand{\abs}[2]{[#1]#2}

\newcommand{\f}[1]{\ensuremath{\mathit{#1}}}

\newcommand{\theory}[1]{\ensuremath{\mathcal{#1}}}

\newcommand{\tunif}[1]{_?{\overset{\theory{#1}}{\approx}}_? }
\newcommand{\unif}{\; _?{\approx}_?\;}

\newcommand{\cunif}{\; _?{\overset{\theory{C}}{\approx}}_?\;}

\newcommand{\context}[1]{\mathbb{C}}

\newcommand{\closenarrow}{\rightsquigarrow_{\theory{R,E}}^c}

\newcommand{\cnarr}[1]{\rightsquigarrow^c_{[#1]}}

\newcommand{\closerewrite}{\longrightarrow_{\theory{R,E}}^{c}}

\newcommand{\starcloserewrite}{\overset{*}{\longrightarrow}_{\theory{R,E}}^{c}}

\newcommand{\teal}[1]{\textcolor{teal}{#1}}

\newcommand{\rew}{\to_\theory{R}}
\newcommand{\erew}{\to_\theory{R,E}}

\newcommand{\ebarrew}{\to_\theory{R/E}}
\newcommand{\plus}{\mbox{\textit{plus}}}
\newcommand{\mult}{\mbox{\textit{mult}}}

\newcommand{\narrow}{\rightsquigarrow}
\newcommand{\cent}{\vdash}

\newtheorem{example}{Example}
\newtheorem{remark}{Remark}
\newtheorem{definition}{Definition}
\newtheorem{theorem}{Theorem}
\newtheorem{lemma}{Lemma}
\newtheorem{proposition}{Proposition}
\newtheorem{corollary}{Corollary}

\journal{\empty}

\begin{document}

\begin{frontmatter}



\title{Nominal Equational Narrowing:\\
Rewriting for Unification in Languages with Binders} %% Article title

\author[kings]{Maribel Fernández}
\author[imperial,unb]{Daniele Nantes-Sobrinho}
\author[unb]{Daniella Santaguida}

\affiliation[kings]{organization={King's College London},%Department and Organization
            % city={},
            country={UK}}

\affiliation[imperial]{organization={Imperial College London},
            country={UK}}

\affiliation[unb]{organization={University of Brasília},
            country={Brazil}}

%% Abstract
\begin{abstract}
Narrowing extends term rewriting with the ability to search for solutions to equational problems. While first-order rewriting and narrowing are well studied, significant challenges arise in the presence of binders, freshness conditions and equational axioms such as commutativity. This is problematic for applications in programming languages and theorem proving, where reasoning modulo renaming of bound variables, structural congruence, and freshness conditions is needed.
 To address these issues, we present a framework for nominal rewriting and narrowing modulo equational theories that intrinsically incorporates renaming and freshness conditions. We define and prove a key property called nominal \theory{E}-coherence under freshness conditions, which characterises normal forms of nominal terms modulo renaming and equational axioms.
Building on this, we establish the nominal \theory{E}-lifting theorem, linking rewriting and narrowing sequences in the nominal setting. This foundational result enables the development of a nominal unification procedure based on equational narrowing, for which we provide a correctness proof. We illustrate the effectiveness of our approach with examples including symbolic differentiation and simplification of first-order formulas.
\end{abstract}

%% Keywords
\end{frontmatter}


%%%%%%%%%%%%%% MAIN TEXT

\section{Introduction}
\label{sec:intro}
In this paper we study the problem of solving equational problems in languages such as Milner's $\pi$-calculus and first-order logic, where the syntax includes  binding operators, freshness conditions and structural congruences. We develop nominal rewriting and narrowing techniques modulo equational axioms, and use them to build  unification procedures.

The nominal framework~\cite{VarBinding/GabbayP02} allows us to specify syntax involving binding and freshness constraints directly within its language. This capability extends beyond what is typically expressible in first-order settings. For example, nominal terms can naturally express identities such as the following from first-order logic:
\begin{equation}\label{eq:fo-identity}
    \forall x. (\phi\wedge \psi)= \phi \wedge \forall x. \psi, \text{ if } x\notin fvar(\phi)
\end{equation}
where $fvar(\phi)$ denotes the set of free variables in the formula $\phi$.
An important distinction to note is that $x$ and $\phi$ are variables at different levels. The symbol $x$ is an object-level variable, representing elements of the domain in a logical interpretation, and is subject to binding. In contrast, $\phi$ is a meta-level variable, representing an entire expression or formula, and is not subject to binding. Nominal techniques carefully distinguish these roles, allowing binding operations to apply selectively. More precisely,  the nominal framework enhances the expressive power of first-order languages by introducing atoms $a$, meta-variables $X$ (called simply variables), and abstractions $[a]t$ of atoms over terms. Atoms represent object-level variables, which can be bound, while variables represent meta-level placeholders, which can be instantiated (substituted) with nominal terms. This clear separation allows us to reason precisely about binding and substitution behaviours.

Identities such as~(\ref{eq:fo-identity}), usually called  {\em schemas}~\cite{GabbayMathijssen06_OneAndaHalfthOrderLogic}, can be naturally expressed in nominal languages. For example, (\ref{eq:fo-identity}) can be expressed as follows:
\begin{equation}\label{eq:id_overview}
    a\#X\vdash \forall [a] (X\wedge Y)\approx X \wedge \forall [a] Y
\end{equation}
where the symbols $\forall$ and $\wedge$ are part of the underlying signature. Here, $X$ and $Y$ are metavariables, $a\#X$ is a freshness constraint stating that $X$ can only be instantiated with terms in which $a$ does not appear free, and equality $\approx$ is $\alpha$-equivalence, hereafter denoted $\aleq$.

%A key strength of the nominal framework is its ability to reason modulo $\alpha$-equivalence while managing freshness conditions explicitly within the syntax. This enables precise and formal manipulation of bindings and substitutions directly within the object language.
Alternative approaches for handling renaming and freshness constraints include the locally nameless representation~\cite{DBLP:journals/jar/Chargueraud12} and higher-order abstract syntax (HOAS)~\cite{DBLP:journals/tocl/McDowellM02}. Note that in higher-order settings, one must reason modulo $\beta$- and $\eta$-equivalence, and in most cases, equational reasoning becomes undecidable.
%, and freshness constraints are typically relegated to the meta-language.
We adopt the nominal approach because nominal equational reasoning remains closer in spirit to first-order reasoning than to higher-order reasoning. A more detailed comparison with these alternative frameworks is provided in the related work section.

%\alert{there is a gap here.} \dani{I moved the unification to the next paragraph}
%MF: Commented out the next paragraph since this has already been said.
%state-of-the-art
%The nominal framework~\cite{VarBinding/GabbayP02} has emerged as a promising approach for dealing with languages involving binders.  In this framework, equality coincides with  $\alpha$-equivalence, denoted $\aleq$, and freshness constraints are integrated within the nominal reasoning rather than being relegated to the meta-language. For example, the expression  $a\# M$  (``$a$ is fresh for $M$'') indicates that if a name $a$ occurs in a term $M$, it must be abstracted by some binder, such as $\lambda$ in the lambda calculus, or $\exists,\forall$-quantification in first-order logic, i.e., $a$  cannot occur free in $M$.

Building on this foundation, \emph{nominal unification} naturally arises as a central operation for equational reasoning in the presence of binders. Unification is fundamental for automated reasoning, serving as the foundation for resolution-based theorem proving, type inference, and numerous other applications.
 Nominal unification~\cite{Matching/jcss/CalvesF10,NomUnification/UrbanPG04} involves finding a substitution $\sigma$ that solves the problem $s\ _?{\overset{}{\approx}}_?\ t$, meaning $s\sigma\aleq t\sigma$, where $s$ and $t$ are nominal terms. Nominal unification algorithms were formalised in proof assistants such as Isabelle \cite{NomUnification/UrbanPG04}, PVS \cite{FormalisingNomC-unif/mscs/Gabriel21} and Coq~\cite{A-C-AC/tcs/Ayala-RinconSFN19}.

In programming languages, theorem proving and protocol analysis applications, unification needs to consider also structural congruences generated by equational axioms $\theory{E}$. %Thus, extensions of nominal unification with equational theories have been investigated.
An approach to solving nominal unification problems modulo equational theories
%(i.e., nominal $\theory{E}$-unification problems),
 involves the use of {\em nominal narrowing}~\cite{NominalNarrowing16}\footnote{Roughly, nominal narrowing is a generalisation of nominal rewriting~\cite{NominalRewriting/FernandezG07,Kikuchi020} by using nominal unification instead of nominal matching in its definition.}. This technique can be used when the equational theory $\theory{E}$ is presented by a convergent nominal rewriting system~\cite{NominalRewriting/FernandezG07,NominalNarrowing16}.
% gap to be investigated
If such a presentation is unavailable,
techniques for rewriting \emph{modulo} \theory{E}~\cite{JouannaudKK83:Incremental,Viola01,VariantNarrowing:EscobarMS09} could be considered.
%impossible.
However, nominal rewriting techniques modulo
\theory{E} have remained unexplored.

%This paper fills this gap by developing nominal \theory{E}-rewriting and narrowing techniques, which we use to build  unification procedures for languages with binding operators and structural congruences. Bridging this gap is illustrated with the dashed lines in Figure~\ref{fig:rew-diagram}.
%prove properties of security protocols.

This paper addresses that gap by introducing a framework for nominal
\theory{E}-rewriting and narrowing, which supports both binding and structural congruence. We use this framework to construct unification procedures capable of reasoning in languages with binding constructs and equational axioms.
Figure~\ref{fig:rew-diagram} summarises this development. Starting from first-order term rewriting systems (TRS) (top), one can either incorporate binding via nominal rewriting (NRS)  (left) or equational axioms via equational term rewriting systems (ETRS) (right). Our work bridges these two extensions by introducing nominal rewriting modulo
$\alpha$-equivalence and
\theory{E}, depicted by the dashed arrows converging at the bottom of the diagram.

% \begin{figure}[!t]
% \begin{center}
% \footnotesize
%     \begin{tikzpicture}[scale=1]
%                 \node (0) at (0,0.5) {\footnotesize First-order (FO) rewriting~\cite{Baader98}};
%                 \node (1) at (0,0) {\theory{R}};
%                 \node (2) at (-1.8,-1.4) {\footnotesize\theory{R} with $\alpha$};
%                 \node (3) at (1.8,-1.4) {\footnotesize\theory{R} with \theory{E}};
%                 \node[left] (4) at (-0.8,-1.9) {\footnotesize Nominal Rewriting~\cite{NominalRewriting/FernandezG07}};
%                 %\node[right] (5) at (0.8,-1.9) {Confluent and Coherent};
%                 \node[right] (6) at (0.8,-1.9) {\footnotesize (FO) Equational TRS~\cite{Jouannaud83:ConfluentandCoherent}};
%                 \node (7) at (0,-3.8) {\footnotesize \textcolor{teal}{
%                 \theory{R} with $\alpha$ and \theory{E}}};
%                \node (8) at (0,-4.2) {\footnotesize \textcolor{teal}{Equational Nominal  Rewriting}};
%                 % \node[left,teal] at (-0.7,-0.5){\scriptsize\(R_1\;\)};
%                 % \node[right,teal] at (0.7,-0.5){\scriptsize\(\;R_2\)};
%                 \draw[->, >=stealth] (1) to (2);
%                 \draw[->, >=stealth] (1) to (3);
%                 \draw[teal,dashed, thick,->, >=stealth] (-1.6,-2.3) to (7);
%                 \draw[teal, thick, dashed,->, >=stealth] (1.4,-2.2) to (7);
%                 \end{tikzpicture}
%     \end{center}
%     \caption{This work: nominal rewriting modulo theories}\label{fig:rew-diagram}
% \end{figure}
\begin{figure}[!t]
\begin{center}
\footnotesize
    \begin{tikzpicture}[scale=.8]
                \node (0) at (0,0.5) {\footnotesize  TRS~\cite{Baader98}};
                \node (1) at (0,0) {\theory{R}};
                \node (2) at (-2.5,-1) {\footnotesize\theory{R} with $\alpha$};
                \node (3) at (2.5,-1) {\footnotesize\theory{R} with \theory{E}};
                \node[left] (4) at (-1.5,-1.5) {\footnotesize NRS~\cite{NominalRewriting/FernandezG07}};
                %\node[right] (5) at (0.8,-1.9) {Confluent and Coherent};
                \node[right] (6) at (1.5,-1.5) {\footnotesize  ETRS~\cite{Jouannaud83:ConfluentandCoherent}};
                \node (7) at (0,-3) {\footnotesize \textcolor{teal}{
                \theory{R} with $\alpha$ and \theory{E}}};
               \node (8) at (0,-3.4) {\footnotesize \textcolor{teal}{Equational Nominal  Rewriting}};

                \draw[->, >=stealth] (1) to (2);
                \draw[->, >=stealth] (1) to (3);
                \draw[teal,dashed, thick,->, >=stealth] (-1.8,-1.9) to (7);
                \draw[teal, thick, dashed,->, >=stealth] (1.8,-1.9) to (7);
                \end{tikzpicture}
    \end{center}
    \vspace{-5mm}
    \caption{This work: nominal rewriting modulo theories}\label{fig:rew-diagram}
\end{figure}
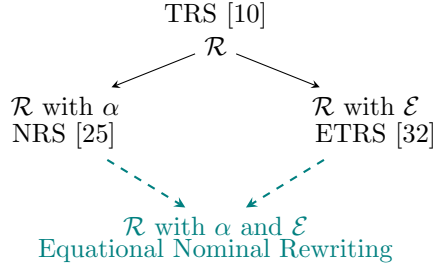

%\paragraph{Contributions.}
%In this work we develop nominal \theory{E}-rewriting and narrowing techniques. %, when a convergent nominal rewrite system equivalent to the theory $\theory{E}$ is unavailable.

Making a parallel to the   first-order setting,  we consider a set \theory{T}  of nominal identities--such as those illustrated in schema~(\ref{eq:id_overview})--and decompose it into a nominal rewriting system $\theory{R}$ and a nominal equational part $\theory{E}$.  This separation underpins two complementary notions of rewriting:
\begin{itemize}
\item In nominal \theory{R/E}-rewriting,  rules from \theory{R} are applied within the equivalence class of a nominal term $t$, modulo $\alpha$ and modulo $\theory{E}$. In other words, in its $(\alpha\cup\theory{E})$-equivalence class.
%class denoted $\approx_{\alpha, \theory{E}}$.
\item In contrast, nominal \theory{R,E}-rewriting relies on ({\em $\theory{E}\cup\alpha$)-matching}, applying rules from \theory{R} directly to the term
$t$ without exploring its entire  $(\alpha\cup\theory{E})$-equivalence class.
\end{itemize}
%description of the developments
 Since the term $t$ may contain variables, the relations \theory{R/E} and \theory{R,E} are subject to freshness conditions.

We show that under a critical property called {\em nominal
\theory{E}-coherence} — which must also respect freshness—these two rewriting relations yield the same normal forms. This result (Theorem~\ref{prop:nom-e-coherence}) generalises the classical coherence theorems from first-order rewriting~\cite{Jouannaud83:ConfluentandCoherent} to the nominal setting.
Extending equational reasoning to the nominal realm is far from trivial:
$(\alpha\cup\theory{E})$-equivalence classes are typically infinite due to binding, and freshness constraints further restrict the set of admissible representatives. We revisit these and other subtleties in detail in \S\ref{ssec:highlight}.

%  We show that the normal form of a nominal term $t$ with respect to \theory{R/E} %(denoted $t{\downarrow_\theory{R/E}}$)
%  corresponds to the normal form of the same term with respect to \theory{R,E}
% %(that is,  $t{\downarrow_\theory{R,E}}$)
% when the relation \theory{R,E} has a property called {\em nominal \theory{E}-coherence} that must work under freshness conditions,  extending results from first-order languages~\cite{Jouannaud83:ConfluentandCoherent} to nominal languages. We highlight that adding equational reasoning into the nominal language is not a trivial extension  as equational theories as ($\alpha\cup \theory{E}$)-equivalence classes are  almost always infinite in the presence of binders and the freshness conditions restricts the representatives in the class. We come back to this later in \S~\ref{ssec:highlight}

Having established the theoretical foundation for nominal rewriting modulo theories, we turn to the practical application of nominal rewriting modulo a theory \theory{E}. To this end, we define the nominal narrowing modulo \theory{E} relation, which operationalises unification (modulo $\theory{R\cup E}\cup \alpha$) by systematically exploring nominal \theory{R,E}-rewriting and instantiation possibilities.

Our approach proceeds in two steps. First, leveraging the nominal \theory{E}-Lifting Theorem (Theorem~\ref{teo:Elifting}), we show that sequences of nominal \theory{R,E}-rewriting steps correspond precisely with sequences of nominal \theory{E}-narrowing steps. This correspondence is the basis of the development of our ($\theory{R\cup E}\cup \alpha$)-unification procedure.
Second, we exploit the notion of {\em closed}\footnote{Closed rules do not have free atoms, a standard restriction in HO rewriting frameworks.} nominal rewriting, which has been shown to be both efficient and complete for equational reasoning~\cite{corr:ClosedNomRew/Fernandez10}.  Inspired by this, we define a closed nominal narrowing relation that yields a sound and complete procedure for nominal ($\theory{R\cup E}\cup \alpha$)-unification (Theorem~\ref{teo:complet_clsd}).

Finally, we observe that, as in the first-order setting, narrowing in our framework may be non-terminating—even when the set of unifiers is finite—due to the inherent complexity of exploring the search space~\citep{VariantNarrowing:EscobarMS09}. To address this challenge, we adapt the  \emph{basic} narrowing strategy to our nominal setting, ensuring derivations are guided and avoid redundant steps. It is well-known that basic narrowing is not complete in the presence of equational axioms~\cite{CLD05_TheFiniteVariantProperty}. Interestingly, our basic nominal narrowing preserves completeness modulo $\alpha$-equivalence. This enables us to define unification procedures for theories $\theory{T} = \theory{R} \cup {\alpha}$, where $\theory{R}$ is a convergent nominal rewrite system that internalises $\alpha$-equivalence. This refinement yields a more tractable--due to termination--approach to nominal $\theory{T}$-unification (Theorem~\ref{Thm:complete_basic}), while remaining faithful to the equational and binding structure of the object language.

%(variant narrowing will be studied in future work).
%developed for first-order languages.
%Since the decidability of nominal \theory{T}-unification relies on the decidability of nominal \theory{E}-unification and \theory{E}-matching, and the only existing nominal unification algorithm is for commutativity \theory{C}, we confirm that nominal \theory{C}-Lifting Theorem holds.
%Finally, due to the volume of extensions that were necessary to establish nominal \theory{E}-narrowing and rewriting, the final goal of using \theory{E}-narrowing as a sound and complete procedure for solving nominal \theory{T}-unification, remains ongoing work.

To conclude this Introduction, we provide a high-level overview of this paper’s contributions, setting the stage for the technical developments that follow in the subsequent sections.

\subsection{Nominal Equational Reasoning At a Glance}\label{ssec:highlight}
This section highlights the distinctive features of nominal rewriting and equational reasoning in the presence of binding, contrasting them with their first-order counterparts through illustrative conceptual examples. These examples demonstrate behaviours unique to nominal languages—phenomena that cannot be directly replicated in first-order frameworks and which have shaped the design of our approach. While the material here is not essential for understanding the technical content of the paper, it serves to build intuition and motivation for the constructions developed later. Readers
%already familiar with nominal techniques may benefit from reviewing this section early, whereas those
less familiar with nominal techniques may prefer to return to this section as needed to clarify concepts introduced in later sections.

%\subsection{Nominal Rewriting}
% }

% A different and very specific behaviour can be observed in nominal terms.  Consider the nominal rewrite system  $R=\{\emptyset\vdash f(X)\to (a\ c)\cdot X\}$, whose right-hand side does not contain ground positions. The rhs is a variable with a suspended renaming.

% For example, consider the term $s_0=f(X)$  and substitution $\rho_0=[X\mapsto h(c)]$. The term $s_0\rho_0=f(h(c))$ is a redex, and the \emph{innermost} reduction is $f(h(c))\to h(a)$ (since $(a \ c)\cdot h(c) =h(a)$). Note that the atom $a$ appears free in the rhs of the derivation. This behaviour cannot be mimicked in first-order languages: a term that appears on the right of a reduction is either introduced via instantiation of a variable and also occurs in the lhs of the reduction, as the term rewrite rules satisfy the conditions on variables; or it is a constant that came from the rhs of the rule. None of these cases express our example.

\paragraph{Nominal Equational Reasoning}

There is a wealth of techniques to reason modulo first-order equational theories, that is, theories defined by a set \theory{E} of first-order axioms — identities of the form $l \approx r$, where $l$ and $r$ are first-order terms (terms without abstractions or atoms). Examples of first-order theories include associativity (\theory{A}), commutativity (\theory{C}), and idempotency (\theory{I}) of operators. For instance, the theory of commutativity for conjunction can be expressed as $\theory{C} = \{ X \wedge Y \approx Y \wedge X \}$.

However, even when $\theory{E}=\emptyset$, the nominal structure yields a richer and more expressive behaviour compared to standard first-order equality:

\begin{example}\label{ex:overview1}
   Consider a signature  $\Sigma=\{\forall, R\}$, where $\forall$ is unary and  $R$ is a binary symbol. First, observe that the $\alpha$-equivalence class of the ground term $\forall [a]R(a,b)$ consists of all terms obtained by renaming the bound atom: %
   $$[\forall [a]R(a,b)]_\alpha=\{\forall [a]R(a,b),\forall [c]R(c,b),\forall [d]R(d,b), \ldots \}.$$
   Now, consider a term with a variable, say $\forall [a]R(a,X)$. The  $\alpha$-equivalence class of this term is not simply obtained by renaming $a$  to any atom: such renamings are valid only if the new atom is fresh for $X$. In particular, $\forall [a]R(a,X)\aleq \forall [c]R(c,X)$ only if $a,c\#X$. Consequently, the larger the freshness context we assume, the larger the $\alpha$-equivalence class might become. For instance, under the context $a,b,c\#X$, we have:
 $$ \forall [a]R(a,X)\aleq \forall [b]R(b,X)\aleq \forall [c]R(c,X).$$
\end{example}

The nominal framework allows us to express richer theories that involve binding and freshness. For example, the theory $\theory{E_{\forall_C}} = \{ \emptyset \vdash \forall[a]\forall[b]X \approx \forall[b]\forall[a]X \}$ expresses that the $\forall$-quantifier is commutative, while the theory $\theory{E_{\forall_\text{e}}} = \{ a\# X \vdash \forall[a] X \approx X \}$ captures the elimination of vacuous quantification, relying on explicit freshness conditions.

%Nominal terms might come with freshness constraints. For example, consider the term-in-context $a\# \lambda [a] f(a,X)$ denotes the fact that $a$ cannot occur free in any instance of $X$. The same with nominal rewriting rules $a\#P \vdash \forall [a]P\to P$.

When $\theory{E} \neq \emptyset$, nominal algebra equality modulo $\theory{E}$ defines what we call an ($\alpha\cup \theory{E}$)-equivalence class of terms. The following example illustrates this point using the equational theory $\theory{E}_{\forall_{\theory{C}}}$, which states that the $\forall$-quantifier is commutative -  a property that cannot be formally captured in standard first-order term languages.

\begin{example}\label{ex:overview2}
Consider the signature $\Sigma= \{\forall, f, >\}$, where $f$ and $>$ are binary symbols. Consider $\theory{E}_{\forall_\theory{C}}= \{\emptyset\vdash \forall[a]\forall[b] X\approx \forall[b]\forall[a] X\}$.
The $\alpha,\theory{E}_{\forall_\theory{C}}$-equivalence class of the term $s= \forall[a]\forall[b]( f(a,b)>0)$ includes
$$
[s]_{\alpha,\theory{E}_{\theory{\forall_\theory{C}}}} = \left\{
    \begin{aligned}
    &\forall[a]\forall[b](f(a,b)>0), \forall[b]\forall[a](f(a,b)>0),\\
    &\forall[c]\forall[b](f(c,b)>0),  \forall[b]\forall[c](f(c,b)>0),\\
    &\forall[c]\forall[d](f(c,d)>0),  \forall[d]\forall[c](f(d,c)>0),\ldots
    \end{aligned}
    \right\}.
$$
Observe that both $\alpha$-renaming of bound variables and the commutativity of quantifiers (via $\theory{E}_{\forall_\theory{C}}$) contribute to generating new representatives of the same equivalence class.
\end{example}

\paragraph{Nominal Rewriting}
In general, computation is driven by rewriting. When a set of identities
$\theory{E}$ guides computational reasoning, the first step is often to attempt to orient these identities into rewrite rules. Given a nominal identity, say $\nabla\vdash l\approx r$,  a nominal rewrite rule~\cite{NominalRewriting/FernandezG07} is obtained by orienting the identity either from left to right, $\nabla\vdash l\to r$ or from right to left, $\nabla\vdash r\to l$. An expected condition for well-formed nominal rewrite rule $\nabla\vdash l \to r$ is that $V(r,\nabla)\subseteq V(l)$ and that $l$ is not a variable. Here, $V(-)$ denotes the set of (meta-)variables of $(-)$.

For example, the identity (\ref{eq:id_overview}) could be oriented in the following nominal rewrite rule: $R=a\#X\vdash \forall [a] (X\wedge Y)\to X \wedge \forall [a] Y$. Note that the name $a$ is representing an object-level variable, therefore, we could have written this rule as
$R^{(a\ b)}=b\#X\vdash \forall [b] (X\wedge Y)\to X \wedge \forall [b] Y$. We write $R^\pi$ to denote the rewrite rule obtained from $R$ by applying the renamings defined by the permutation $\pi$. We call a set of rewrite rules {\em equivariant} when it is closed under the action of permutations $(-)^{(a\ b)}$ for all atoms $a$ and $b$. When we write our nominal rewrite rules, we do not give all the possible variants, but we work implicitly with the {\em equivariant closure} $\theory{R}$  of the rewrite rules.

To apply a rewrite step in a term  $\forall [c](R(a,Z)\wedge \neg(R(a,c)) )$ we need to match it with the lhs of some possibly renamed version of the rule $R$. If $c\#Z$ then we have a step $\forall [c](R(a,Z)\wedge \neg(R(a,c)) ) \to R(a,Z)\wedge \forall [c](\neg(R(a,c)) )$.
Checking that a nominal rule $R$ matches a term $s$  is a polynomial problem~\cite{Matching/jcss/CalvesF10}, but checking whether $s$ matches $R^\pi$ for some $\pi$ is an NP-complete problem in general~\cite{DBLP:conf/icalp/Cheney04}.  For efficiency we are interested in conditions to make this problem polynomial,
for this reason, we will work with {\em closed} systems, which intuitively corresponds to working with rules without free atoms.
We come back to this in \S\ref{sec:Narr-for-Unif}.

Closed systems have many applications. All systems that stem from functional programming, including the axiomatisation of the $\lambda$-calculus, are closed. Additionally, all systems that can be specified within a standard higher-order rewriting formalism are also closed~\citep{FG05_NomRewWithNameGeneration}.

%%%%%%%%%%%%%%%%%%%%%%%

\paragraph{Nominal Rewriting Modulo a Theory}

Some theories, like commutativity, may not be oriented without losing important properties such as termination. To overcome this situation, we may
decompose our arbitrary theory into two different sets: one ordered (the usual set of rules), and another unordered (the identities). For instance, consider the theory \theory{T} defined by the equivariant closure of  the identities:
%
%This property is a result of  P /\ P -> P and it is independent of the abstraction. I will modify the rule \emptyset\vdash \forall[a]P\wedge \exists[a]P\to \exists[a]P
$$
\left\{
\begin{aligned}
a\#X&\vdash \forall [a]X\approx X,\\
%&a\#Y \vdash \forall[a](X\wedge Y) \approx (\forall[a]X)\wedge Y
\emptyset& \vdash \forall[a]\forall [b] Y \approx \forall [b]\forall [a]Y
\end{aligned}\right\}
$$
These identities come from a well-known fragment of first-order logic and describe standard properties of quantifiers.   For simplicity, we do not consider the commutativity of
$\wedge$ or $\vee$ here, though it could be included as additional identities if needed.

The first identity is best oriented from left-to-right to avoid generating infinite equivalence classes due to expansion. Note that the equivalence class of a term such as $\forall[b]b$  would be infinite, including, for example, terms like  $\forall[c](\forall[b]b),~  \forall[a](\forall[c](\forall[b]b))$  and so on—not merely due to
$\alpha$-equivalence, but also due to repeated, semantically vacuous quantification. Left-to-right orientation prevents this unbounded expansion.% while preserving intended semantics.

In contrast, the second identity naturally lends itself to reasoning within an equivalence class that is already infinite due to
$\alpha$-renaming. However, the application of the identity itself (i.e., commuting $\forall$-abstraction) does not introduce (additional)  structural infinite variability beyond that already implied by
$\alpha$-equivalence. As such, it can be safely treated as part of the unordered equational component.

Following this reasoning, we decompose  \theory{T} as \theory{R{\cup}E} where  \theory{R} is obtained by the equivariant closure of the rule
$R=\{
 a\#X\vdash \forall [a]X\to X
\}$ and  \theory{E} is the equivariant closure of the schema
$\{\emptyset \vdash \forall[a]\forall [b] Y \approx \forall [b]\forall [a]Y\}$. We can verify that \theory{R}  is terminating and confluent modulo $\theory{E}$. We will come back to this in \S\ref{sec:e-rewriting}, see Example~\ref{exa:rew-modulo}.

Nominal rewriting modulo an equational theory is inherently more involved than its first-order counterpart. However, it remains significantly more tractable than analogous mechanisms in full higher-order rewriting frameworks, striking a balance between expressiveness and manageability.

\paragraph{Unification Modulo a Theory via Closed Narrowing}
As expected,  nominal \theory{R,E}-narrowing is a generalisation of nominal \theory{R,E}-rewriting which uses nominal \theory{E}-unification instead of nominal \theory{E}-matching, i.e., both the left-hand side of the rule and the term to be rewritten can be instantiated. Nominal \theory{R,E}-narrowing will be used as a strategy for solving nominal unification problems modulo a theory \theory{T}. The idea is to decompose
\theory{T} as \theory{R{\cup}E}
where
\theory{E} is a set of non-oriented identities for which a unification algorithm exists and  \theory{R} is a set of rewriting rules which is \theory{E}-convergent (see Definition~\ref{def:RE-rewriting}).

For example, consider the set $\theory{T}$ of identities:
$$\{c\#P \vdash P\land \forall[c]Q \approx \forall[c](P\land Q), ~\emptyset \vdash \forall[a]\forall[b]Y \approx \forall[b]\forall[a]Y\}$$
and the decomposition into rules
% consider the oriented identity~\ref{eq:id_overview} as the rule in
$\theory{R}=\{c\#P \vdash P\land \forall[c]Q \to \forall[c](P\land Q)\}$ and identities $
\theory{E}= \{\emptyset \vdash \forall[a]\forall[b]Y \approx \forall[b]\forall[a]Y\}.$
In order to solve the \theory{T}-unification problem between  $s_0=\forall[c](P_1\land\forall[c']Q_1)$ and $t_0=\forall[c'](P_2\land\forall[d]Q_2)$, each with their own freshness constraints:
$$(c'\#P_1 \vdash s_0) \;\tunif{T}\; (c\#P_2\vdash t_0),$$
we apply sequences of \theory{R,E}-narrowing narrowing steps (denoted $\narrow_\theory{R,E}$) on the two terms in parallel, until we obtain two new terms and that are \theory{E}-equivalent under a possibly new freshness context $\Delta_n$. That is,
%$\Gamma =\{c'\#P',c\#P'\}$:

$$\{c'\#P',c\#P'\}\vdash (s_0,t_0) \narrow_\theory{R,E} \ldots \narrow_\theory{R,E} \Delta_n\vdash (s_n,t_n) $$
where $\Delta_n\vdash s_n\ealeq t_n$. This way, we can guarantee that the initial terms are \theory{T=R{\cup}E}-equiv\-a\-lent. We will come back to this procedure with a more interesting example in \S\ref{ssec:exam_diff}.

% $$
% \begin{aligned}
% c'\#P_1\vdash \forall[c](P_1 \land \forall[c']Q_1) &\closenarrow \cdots &\closenarrow c'\#P'\vdash\forall[c]\forall[c'](P'\land\forall[d]Q')\\
% &&\ealeq \\
% c\#P_2\vdash \forall[c'](P_2\land \forall[d]Q_2) &\closenarrow \cdots & \closenarrow c\#P' \vdash \forall[c']\forall[c](P'\land\forall[d]Q')
% \end{aligned}
% $$

% $(t,s) => ... narrowing steps => (t',t'')$
% where $t'=_e t''$.

% \dani{Solution: $\theta=[P_1\mapsto P', P_2\mapsto P', Q_1\mapsto \forall[d]Q', Q_2\mapsto \forall[c]Q', Q\mapsto \forall[d]Q']$}

To establish the correctness of our nominal unification procedure via narrowing, we first prove a key result that connects
\theory{R,E}-rewriting sequences to  \theory{R,E}-narrowing sequences. Specifically, we show how one can {\em lift} a derivation from rewriting to narrowing, and vice versa (Theorem~\ref{teo:Elifting}). To prove the completeness of our unification procedure modulo
\theory{T}, we show that this correspondence also holds in the closed setting (Theorem~\ref{teo:complet_clsd}).

%{\bf Contribution 2. } We formulate and prove the {\em nominal \theory{E}-Lifting Theorem} (Theorem~\ref{teo:Elifting}) which establishes a precise correspondence between nominal \theory{R,E}-narrowing   and nominal \theory{R,E}-rewriting.

% {\bf Contribution 3.} We define a sound and complete nominal  \theory{T}-unification procedure (Theorem~\ref{teo:complet_clsd}), for the case where the nominal equational theory $\theory{T=R{\cup}E}$ is  \emph{closed}.

%  \item Applications of nominal equational narrowing, including symbolic analysis of security protocols.

\paragraph{Refinements and Challenges}
As in the
first-order approach, nominal narrowing may be non-terminating, meaning that this (\theory{R\cup E \cup \alpha})-unification procedure might not terminate, even when only a finite number of unifiers exist (see Example~\ref{exa:narrow-infty}).
%
%One of the prices to pay for nominal narrowing is that our derivation tree may be infinite in depth (infinite sequence of narrowing steps) and in width (infinite solutions from a fixed-point problem). Fortunately, our \theory{T}-unification procedure still works if we enumerate all the nodes of the narrowing tree. However, it is necessary to seek refinements to nominal narrowing in order to possibly eliminate non-termination.
%
When \theory{E=\emptyset}, refinements were proposed in first-order approaches for the restricted case when \theory{R\cup\emptyset}-unification is finitary and the narrowing relation induced by \theory{R} is terminating~\citep{Hullot80,KN87_matching_unif_complexity,DMS92_DecidableMatchingConvergentSystems,Mitra94,DBLP:journals/jsc/NuttRS89}. However, in the nominal framework, even when $\theory{E}=\emptyset$, we are always working modulo the $\alpha$-equivalence theory. Therefore, the usual first-order approach does not directly translate to our framework.
 In \S\ref{sec:basnarr} we prove that {\em closed} basic narrowing is terminating and a complete unification algorithm exists~(Theorem~\ref{Thm:complete_basic}).

%{\bf Contribution 4.} Basic closed narrowing is used to prune the nominal narrowing derivation tree. This pruning enables the construction of a complete unification algorithm for unification modulo \theory{T}  .

\smallskip
%In the face of this, our first refinement is basic closed nominal narrowing, which is used to eliminate redundant closed narrowing derivations to give sufficient conditions for the desired termination. Roughly speaking, in a basic derivation, narrowing is never applied to a subterm introduced by a previous narrowing substitution.

In exploring alternative refinements, we focused on the behaviour of nominal narrowing with respect to a non-empty equational theory. This presents several challenges. First, we are dealing with infinite equivalence classes. Second, it is well known that first-order narrowing modulo
\theory{E} can generate infinite derivation sequences—particularly in the presence of associativity and commutativity axioms—as illustrated in~\citep{CLD05_TheFiniteVariantProperty,Viola01}. Such behaviour compromises both termination and, in some cases, the completeness of
\theory{E}-unification via narrowing (see Example~\ref{exa:incomplete-ac}).

Most of the alternative approaches to first-order unification modulo \theory{T} using \theory{R,E}-narrowing techniques rely on restricting the equational theory \theory{E} to satisfy some requirements including the existence of \theory{E}-unification algorithms that produce a finite, minimal, and complete set of unifiers. Unfortunately, in the nominal setting, unification algorithms are currently only available for very limited theories—most notably, commutativity. As a result, further refinements cannot be pursued without first advancing the state of the art in nominal unification techniques.

This gap highlights a significant opportunity for future research and motivates our decision to defer the development of such refinements to future work (see \S\ref{ssec:basic_equational}).

\subsection{Organisation of the paper}
This paper is a revised and extended version of previous work~\cite{NominalNarrowing16,LSFA2024nominal,UNIF2024nominal}
% \daniele{Perhaps you can cite the pre-proceedigns of lsfa and unif. You need to create an bib entry for that by hand.}
where nominal narrowing and rewriting were extended to consider equational axioms. Here we present also the unification procedure, describe applications, and include proofs  previously omitted.

In \S\ref{sec:preliminaries} we present  preliminary definitions to make the paper self-contained. In \S\ref{sec:e-rewriting} and \S\ref{sec:e-narrowing} we extend the notions of rewriting and narrowing, respectively, modulo an equational theory \theory{E} to the nominal framework and prove the Lifting Theorem  (\S\ref{section:NomLiftingModuloC}), which is the basis for  proving completeness of the unification procedure presented in \S\ref{sec:Narr-for-Unif}.
In \S\ref{sec:basnarr} we present the closed basic narrowing refinement and discuss the current limitations to extend nominal narrowing techniques to work modulo theories.
% Examples of applications and alternative narrowing strategies are given in \S\ref{sec:appl}.
Finally, we discuss related work in \S\ref{sec:rel-work} and conclude in
\S\ref{sec:future-work}.

\section{Preliminaries}\label{sec:preliminaries}

We briefly recall the formal definitions of nominal rewriting and unification. For more details, we refer to~\citep{NominalRewriting/FernandezG07}.
%Already said??
%We will use $\equiv$ for syntactic equality, $=$ for definitions and $\approx_\alpha$ for $\alpha$-equality.

\subsection{Nominal Syntax}
Fix countable infinite  disjoint sets of {\em atoms} \(\mathbb{A} = \{a,b,c,\ldots \}\) and {\em variables} \({\cal X} = \{X, Y, Z, \ldots \}\).
Atoms follow the \textit{atom convention}: atoms \(a, b, c,\ldots \)  over \(\atomSet\) represent different names.
Let \(\Sigma\) be a finite set of term-formers disjoint from \(\mathbb{A}\) and \({\cal X}\) such that  each \(f \in \Sigma\) has a given arity (a non-negative integer \(n\)), written $f:n$.
A \textit{permutation} \(\pi\) is a bijection on \(\atomSet\) with finite domain, i.e., the set \(\dom{\pi} = \{a \in \atomSet \mid \pi(a) \neq a \}\) is finite. The identity permutation is denoted $id$. The composition of permutations $\pi$ and $\pi'$ is denoted $\pi\circ \pi'$ and $\pi^{-1}$ denotes the inverse of the permutation $\pi$.

{\em Nominal terms} are defined  inductively by the grammar:
 $$ s,t,u ::= a \mid \pAction{\pi}{X} \mid  \abs{a}{t} \mid f(t_1, \ldots, t_n),$$
where \(a\) is an {\em atom}, \(\pAction{\pi}{X}\) is a (moderated) variable, \(\abs{a}{t}\) is the {\em abstraction} of \(a\) in the term \(t\), and \(f(t_1, \ldots, t_n)\) is a {\em function application} with   \(f:n \in \Sigma\).  We abbreviate $id\cdot X$ as $X$. A term is \textit{ground} if it does not contain variables.
A \emph{position} $\context{C}$ is defined as a pair $(s, \_)$ of a term and a distinguished variable $\_ \in \mathcal{X}$ that occurs exactly once in $s$\footnote{This notion of position is standard for nominal terms and equivalent to the notion of position as a path in a term.}.
We write $\context{C}[s']$ for $\context{C}[\_ \mapsto s']$ and if $s\equiv \context{C}[s']$, we say that $s'$ is a subterm of $s$ with position $\context{C}$. The root position will be denoted by $\context{C} = [\_]$.
For example, if $s = f([a]X,b)$, the set of positions of $s$ is $\{[\_],f([\_],b),f([a][\_],b),f([a]X,[\_])\}$.

The \emph{permutation action} of $\pi$ on a term $t$ is defined by induction on the term structure as follows:
$
%$
% \begin{array}{cc}
% \begin{array}{rcl}
\pi\cdot a = \pi(a),
\pi\cdot [a]t = [\pi\cdot a](\pi\cdot t),
%\end{array}
%&
%\begin{array}{rcl}
 \pi \cdot (\pi'\cdot X) = (\pi \circ \pi')\cdot X,
 \pi\cdot f(t_1,\ldots, t_n) = f(\pi\cdot t_1, \ldots, \pi\cdot t_n).
%\end{array}
%\end{array}
%$
$
The \emph{difference set} of two permutations is denoted
$ds(\pi,\pi') = \{ n \; | \; \pi\cdot n \neq \pi'\cdot n \}$, and $ds(\pi,\pi')\# X$ represents the set of constraints $\{ n\# X \; | \; n \in ds(\pi,\pi') \}$. For example, if $\pi = (a \ b)(c \ d)$ and $\pi' = (c \ b)$, then $ds(\pi,\pi') = \{a,b,c,d\}$ since $\pi$ and $\pi'$ act differently in each atom: note that $\pi(a)=b$ and $\pi'(a)=a$. In addition,  $ds(\pi,\pi')\# X = \{a\#X,b\#X,c\#X,d\#X\}$.

The \emph{substitution action} of $\theta$ on a term $t$ is defined by induction on the term structure:
$
%$
%\begin{array}{cc}
%\begin{array}{rcl}
a\theta = a,
([a]t)\theta = [a](t\theta),
%\end{array}
%&
%\begin{array}{rcl}
 (\pi\cdot X)\theta = \pi\cdot (X\theta),
 f(t_1,\ldots, t_n)\theta = f(t_1\theta, \ldots, t_n\theta).
%\end{array}
%\end{array}
$
%$
The domain of a substitution $\theta$ is written as $\texttt{dom}(\theta)$, and the image  $\texttt{Im}(\theta)$. Therefore, if $X \not\in \texttt{dom}(\theta)$ then $X\theta = X$. The restriction of the domain to a set $V\subseteq {\cal X}$ of variables generates the substitution $\theta|_V$,  the \emph{restriction of $\theta$ to $V$}.
The identity substitution is denoted {\tt Id}.  The composition of two substitutions $\theta_1$ and $\theta_2$ is denoted by simple juxtaposition
as $\theta_1\theta_2$ and  $t\theta_1\theta_2=(t\theta_1)\theta_2$.

\subsection{Nominal Constraints and  Judgements}
In the nominal framework, there are two kinds of constraints:  (i) $s\aleq t$ is an (alpha-)equality constraint, which means that $s$ and $t$ are equal up to the renaming of bound names; (ii) $a\#t$ is a freshness constraint which means that $a$ cannot occur unabstracted in $t$. A freshness constraint of the form $a\#a$ is called {\em inconsistent} and a freshness constraint of the form $a\#X$ is called {\em primitive}.

A freshness {\em context}, denoted $\nabla,\Delta$, consists of finite sets of primitive freshness constraints. We  abbreviate $\{a\# X, b\# X, c\# X\}$ as $a,b,c\# X$. {\em Judgements} have the form $\Delta\vdash s\aleq t$ and $\Delta\vdash a\# t$ and are derived using the rules in Figure~\ref{fig:fresh-and-equalrelation}.
Given a finite set $Pr$ of (freshness or equality) constraints and a context $\Delta$, we write $\Delta\vdash Pr$ if $\Delta\vdash C$ can be derived, for each constraint $C\in Pr$.

    \begin{figure}[!t]
%\resizebox{10cm}{!}{
   %\begin{mdframed}
    % \centering
    \footnotesize{
    \begin{mathpar}
    \inferrule*[right={\scriptsize (\#a)}]{\ }{\Delta \vdash a \# b}
%%%
 \and
    \inferrule*[right={\scriptsize (\#x)}]{(\pi^{-1} \cdot a \# X) \in \Delta}
    {\Delta \vdash a \# \pi \cdot X} ~~
    \and
%%%
    \inferrule*[right={\scriptsize (\#aba)}]{\ }{\Delta \vdash a \# [a]t}
%%%
%
\and
    \inferrule*[right=\scriptsize{(\#abb)}]{\Delta \vdash a \# [b]t }{\Delta \vdash a \# t}
%%%
    \and
%%%
    \inferrule*[right={\scriptsize(\#f)}]{\Delta \vdash a \# t_1 \\  \cdots \\ \Delta \vdash a \# t_n}{\Delta \vdash a \# f(t_1, \cdots, t_n)} ~~
%%%
\and
\inferrule*[right={\scriptsize ($\aleq$a)}]{\ }{\Delta \vdash a\aleq a}
%%%
  \and
    \inferrule*[right={\scriptsize ($\aleq$x)}]{ds(\pi,\pi')\# X \in \Delta }{\Delta\vdash \pi\cdot X \aleq \pi'\cdot X} ~~
%%%
    \and
%%%
    \inferrule*[right={\scriptsize ($\aleq$f)}]{\Delta \vdash s_1 \aleq t_1 ~  \cdots  \Delta \vdash s_n\aleq t_n }{\Delta \vdash f(s_1, \cdots, s_n) \ \approx_\alpha \ f(t_1, \cdots, t_n)} ~~
%%%
%
\and
    \inferrule*[right={\scriptsize ($\aleq$aba)}]{\Delta\vdash s\aleq t}{\Delta\vdash [a]s\aleq [a]t} ~~
%%%
\and
    \inferrule*[right={\scriptsize ($\aleq$abb)}]{\Delta\vdash s\aleq (a\ b)\cdot t ~~ \Delta\vdash a\# t}{\Delta\vdash [a]s\aleq [b]t} ~~
    \end{mathpar}
    }
   % \end{mdframed}
%}
    \caption{Rules for $\#$ and $\approx_{\alpha}$}\label{fig:fresh-and-equalrelation}
    \end{figure}

\begin{example}\label{ex:deriv}
Consider the signature {\normalfont $\Sigma_{\lambda}=\{\texttt{lam}:1, \texttt{app}:2\}$} for the lambda-calculus.
Let {\normalfont $Pr=\{\texttt{lam}[a]\texttt{app}(a,X)\aleq \texttt{lam}[b]\texttt{app}(b,(a \ c)\cdot X)\}$} be a problem and  $\Delta = \{a,b,c\#X\}$ be a context. The judgement {\normalfont $\Delta\vdash \texttt{lam}[a]\texttt{app}(a,X)\aleq \texttt{lam}[b]\texttt{app}(b,(a \ c)\cdot X)$} is derivable using the rules in Figure~\ref{fig:fresh-and-equalrelation}.

\end{example}

Given a freshness context $\Delta$ and a substitution $\theta$, $\Delta\theta$ consists of the set of  freshness constraints $\{a\#X\theta \mid a\#X \in \Delta\}$.  We denote by $\nfpair{\Delta\theta}$ the least context $\Delta'$ such that $\Delta'\vdash \Delta\theta$, if such context exists. If such context does not exist,  $\Delta\theta$ contains a constraint that cannot be derived using the rules in Figure~\ref{fig:fresh-and-equalrelation}. For example, for $\Delta=\{a\#X,b\#Z\}$, $\theta=[X\mapsto f(a,a),Z\mapsto c]$ and $\Delta\theta=\{a\#f(a,a), b\#c\}$ there exists no $\Gamma$ such that $\Gamma\vdash a\# f(a,a)$. Thus, $\nfpair{\Delta\theta}$ does not exist. Differently, for  $\theta'=[X\mapsto f(b,Y),Z\mapsto c]$, we have $\nfpair{\Delta\theta'}=\{a\#Y\}$.

A {\em term-in-context} $\Delta\vdash t$ expresses that the term $t$ has the freshness constraints imposed by $\Delta$. For example, $a\#X\vdash f(X,h(b))$ expresses that $a$ cannot occur free in instances of $X$.

\subsection{Equality Modulo an Equational Theory \theory{E}}\label{ssec:theories}
A nominal {\em identity}  is a  pair in-context, $\nabla\vdash (l,r)$, of nominal terms $l$ and $r$ under a (possibly empty) freshness context $\nabla$.  We denote such identity as $\nabla \vdash l\approx r$.

Given an equivariant set \theory{E} of identities closed by symmetry, the  induced one-step equality is generated by:
{\small
\begin{mathpar}
    \inferrule*[right=($Ax_\theory{E}$)]{\Delta, \Gamma \vdash \big(\nabla\theta, \\  s\approx_\alpha \context{C}[(l\theta)], \\ \context{C}[(r\theta)]\approx_\alpha t \big)}
    {\Delta \vdash s\approx_{\alpha,\theory{E}} t}
\end{mathpar}
}

\noindent for any position $\context{C}$,
%permutation $\pi$,
substitution $\theta$, and fresh context $\Gamma$ (so if $a\# X \in \Gamma$ then $a$ is not mentioned in $\Delta, s, t$).

The {\em equational theory} induced by \theory{E}, which we will also denote as \theory{E}, is the least transitive reflexive closure of  the one-step equality.

In other words, an identity from $\theory{E}$ (possibly after renaming) can be applied at any position  (via $\context{C}$) and with any substitution $\theta$, provided freshness conditions are respected.
Notice that $\alpha$-equality is naturally embedded into this framework, since the equations in $\theory{E}$ can include abstractions\footnote{In the case where $\theory{E}$ is purely first-order (i.e., no binding structure), the freshness context $\nabla$ and the equivariant closure become superfluous and can be omitted.} and we use $\alpha$-equality in the premises of $(Ax_\theory{E})$.  From now on, to simplify the notation, we will write $\ealeq$ instead of $\approx_{\alpha,\theory{E}}$.

The following proposition is a direct consequence of the definition of equational theory.

\begin{proposition}[\theory{E}-Compatibility with substitutions]\label{def:e-compatible}
    Equational theories
    %\theory{E} is
    are {\em compatible with  substitutions}, that is,  whenever $\pair{\Delta\theta}_{nf}$ exists, if $\Delta \vdash Pr$ then $\pair{\Delta\theta}_{nf} \vdash Pr\theta$.
%     \begin{enumerate}
%     \item If $\Delta \vdash a\# t$ then $\pair{\Delta\theta}_{nf} \vdash a\# (t\theta)$.
%     \item If $\Delta \vdash s \ealeq t$ then $\pair{\Delta\theta}_{nf} \vdash (s\theta) \ealeq (t\theta) $.
%     \item If $\Delta \vdash Pr$ then $\pair{\Delta\theta}_{nf} \vdash Pr\theta$.
% \end{enumerate}
\end{proposition}

The next result which will be used in the proofs is a direct consequence of Proposition~\ref{def:e-compatible} and follows by inspection of the rules in Figure~\ref{fig:fresh-and-equalrelation}.

\begin{lemma}\label{lem:compatibility}
    Consider the contexts $\Delta,\nabla,\Gamma$ and the substitutions $\sigma,\theta$. If we have $\nabla \vdash \Delta\sigma$ and $\Gamma \vdash \nabla\theta$ then $\Gamma \vdash \Delta\sigma\theta$.
\end{lemma}

\subsection{ Nominal \theory{E}-unification}

We now introduce the basic definitions for nominal unification modulo an equational theory $\theory{E}$, which will be used throughout the paper. A substitution $\theta_1$ is said to be {\em more general} than a substitution $\theta_2$ if there exist a substitution $\sigma$ and a context $\Delta$ such that, for every variable $X \in \mathcal{X}$, we have $\Delta \vdash X\theta_2 \ealeq X\theta_1\sigma$. This notion naturally extends to pairs $(\Delta_1, \theta_1)$ and $(\Delta_2, \theta_2)$, where substitutions may involve variables subject to freshness constraints and where the equational theory $\theory{E}$ is taken into account. Specifically, we say that $(\Delta_1, \theta_1)$ is more general than $(\Delta_2, \theta_2)$, written $(\Delta_1, \theta_1) \leq_{\theory{E}} (\Delta_2, \theta_2)$, if there exists a substitution $\theta'$ such that $\Delta_2 \vdash X\theta_1\theta' \ealeq X\theta_2$ for all $X \in \mathcal{X}$, and $\Delta_2 \vdash \Delta_1\theta'$. We denote by $\leq^V_{\theory{E}}$ the restriction of $\leq_{\theory{E}}$ to a given set $V$ of variables. Algorithms for (nominal) unification can be defined based on simplification rules as seen in \citep{Baader98,NomUnification/UrbanPG04,NominalRewriting/FernandezG07}.

\begin{definition}\label{def:nominalEunif}
A \emph{nominal $\theory{E}$-unification problem (in-context)} $\uP$ is a set of equations of the form
 $(\nabla\vdash l)\;_?{\overset{\theory{E}}{\approx}}_?  \; (\Delta\vdash s)$.
The pair $(\Delta',\theta)$ is an $\theory{E}$-solution (or $\theory{E}$-unifier) of $\uP$ iff
$(\Delta',\theta)$ solves each  equation  in $\uP$ %$\{\nabla,\Delta, l\eunif s \}$,
that is, for each such $(\nabla\vdash l) \;_?{\overset{\theory{E}}{\approx}}_?  \; (\Delta\vdash s)$,
the following hold:
\begin{itemize}
\item  $\Delta'\vdash \nabla\theta, \Delta\theta$; and
\item  $\Delta'\vdash  l\theta\ealeq s\theta$.
\end{itemize}
The set of all $\theory{E}$-solutions $\uP$ is denoted as $\mathcal{U}_{\theory{E}}(\uP)$.
A subset $\mathcal{V}\in \mathcal{U}_{\theory{E}}(\uP)$ is  a \emph{complete set of $\theory{E}$-solutions of $\uP$} if for all $(\Delta_1,\theta_1) \in \mathcal{U}_{\theory{E}}(\uP)$, there exists $(\Delta_2, \theta_2) \in \mathcal{V}$ such that $(\Delta_2,\theta_2) \leq_{\theory{E}} (\Delta_1, \theta_1)$. In addition, $\mathcal{V}$ is \emph{minimal} iff $\forall (\Delta_1,\theta_1), (\Delta_2,\theta_2) \in \mathcal{V}$, both  $(\Delta_1,\theta_1) \not\leq_\theory{E} (\Delta_2,\theta_2)$ and $(\Delta_2,\theta_2) \not\leq_\theory{E} (\Delta_1,\theta_1)$.
\end{definition}

\begin{example}\label{ex:c-unification} Let $\Sigma = \{h:1, f^\theory{C}:2, \oplus: 2\}$ be a signature, where $f^\theory{C}$ and $\oplus$ are commutative symbols, i.e., and $\theory{C} =\{ \ \vdash f^\theory{C}(X,Y) \approx f^\theory{C}(Y,X), \ \vdash X\oplus Y \approx Y\oplus X\}$ be the axioms defining the theory.
Consider the \theory{C}-unification problem:
\begin{equation}\label{eq:1}
    (\emptyset\vdash f^\theory{C}([a][b]Z,Z)) \; _?{\overset{\theory{C}}{\approx}}_? \; (\emptyset\vdash  f^\theory{C}([b][a]X,X))
\end{equation}
% $$(\emptyset\vdash f^\theory{C}([a][b]Z,Z)) \; _?{\overset{\theory{C}}{\approx}}_? \; (\emptyset\vdash  f^\theory{C}([b][a]X,X))$$
% Using the Nominal \theory{C}-unification algorithm~\citep{FormalisingNomC-unif/mscs/Gabriel21},
% % we get:
% it can be simplified to:
Using a standard \theory{C}-decomposition rule, this problem branches to either
$$\text{(i)}\; \{[a][b]Z \cunif [b][a]X, Z\cunif X\} \qquad \text{(ii)}\; \{[a][b]Z\cunif X, Z\cunif [b][a]X\}$$
The interesting branch is $\text{(i)}$. A simple instantiation rule transforms the problem in $\text{(i)}$ into the problem $\{[a][b]X \cunif [b][a]X\}$ using the substitution $\sigma=[Z\mapsto X]$.
% This problem has infinite solutions. In fact:
The problem~\ref{eq:1} has a solution:
$(\{a\#X,b\#X\}, \sigma\rho_1)$ with instances $\rho_1$ of $X$ that do not contain free occurrences of $a$ or $b$. E.g. for $\rho_1=[X\mapsto g(e)]$, we have
$[a][b]X\rho_1 \caleq [b][a]X\rho_1$.
However, there are infinite solutions considering a \theory{C}-function symbol whose arguments contain $a$ and $b$. In fact,
\begin{itemize}
    \item $(\emptyset, \rho_2
= \sigma[X\mapsto a \oplus b])$: since $X\rho_2 = a\oplus b \caleq b \oplus a = (a\ b)\cdot X\rho_2$
    \item $(\emptyset, \rho_3
= \sigma[X\mapsto (a \oplus b)\oplus(a \oplus b)])$: since
$X\rho_3 = (a \oplus b)\oplus(a \oplus b) \caleq (b \oplus a)\oplus(b \oplus a) = (a\ b)\cdot X\rho_3.$ And so on and so forth.
\end{itemize}

\end{example}

\section{Nominal Rewriting Modulo Theories} \label{sec:e-rewriting}

An {\em equational nominal rewrite system} (ENRS) is a set of nominal identities $\theory{T}$ that can be partitioned into two components: a set $\theory{R}$ of oriented nominal rewrite rules--possibly featuring freshness constraints such as $\nabla\vdash l\to r$--and a set $\theory{E}$ of unordered (equational) identities. We denote this decomposition as $\theory{R} \cup {\alpha} \cup \theory{E}$, or more concisely, $\theory{R} \cup \theory{E}_\alpha$.

In the nominal framework, we work modulo $\alpha$-equivalence, meaning all terms are considered up to $\alpha$-renaming.  With an ENRS, we generalise this by reasoning modulo combined $\alpha,\theory{T}$-equivalence classes, where part of $\theory{T}$ may be oriented into a set of nominal rewrite rules ($\theory{R}$), while the rest ($\theory{E}$) remains as equational axioms. This combination gives rise to the composite rewriting relation: $\ealeq \circ \to_\theory{R}\circ \ealeq$, whose step-reduction is defined below.
As said before, we work implicitly with the {\em equivariant closure} of rewrite rules $\theory{R}$.

When $\theory{E} = \emptyset$, this reduces to plain nominal rewriting $\aleq \circ \to_\theory{R} \circ \aleq$, as introduced in~\cite{NominalRewriting/FernandezG07}:
%Let \theory{R} be a set of nominal rewriting rules. Then, $\to_\theory{R}$ is defined as:
{\small
\begin{mathpar}
\inferrule{ s\equiv\context{C}[s'] \\ \Delta \vdash \nabla \theta, \\  s'\aleq l\theta, \\ t \aleq \context{C}[ r \theta]}{\Delta\vdash s \to_\theory{R} t}
\end{mathpar}
}

\noindent for a substitution $\theta$,
% permutation $\pi$,
subterm $s'$ of $s$, position $\context{C}$ and rule $\nabla\vdash l\to r\in \theory{R}$.
As expected, to find the substitution
% and permutation
above to build a rewriting step, we need to solve a (nominal) matching problem.
As usual, we will write $\to^+$ to denote at least one-step reduction, and $\to^*$ to denote a finite (possibly zero) number of reductions.

In this paper, we are interested in the case $\theory{E}\neq \emptyset$, and the rewrite relation is defined in $(\alpha,\theory{E})$-equivalence classes:

\begin{definition}[Nominal \theory{R/E}-rewriting]\label{def:RE-rewriting}
Let
% $\theory{R\cup E}$
$\theory{R}{\cup}\theory{E}_\alpha$
be an ENRS. The relation $\to_\theory{R/E}$ is induced by the composition $\ealeq \circ \to_\theory{R}\circ \ealeq$. A nominal term-in-context $\Delta \vdash s$ reduces with   $\to_{\theory{R/E}}$, when a term in its $\theory{E}$-equivalence class reduces via $\to_{\theory{R}}$ as below:
\begin{center}
$\Delta \vdash (s \to_{\theory{R/E}} t)$ iff there exist $s',t'$ such that $\Delta \vdash (s \ealeq s' \to_\theory{R} t' \ealeq t)$.
\end{center}
We say  that $\theory{R}$ is \emph{$\theory{E}$-confluent} if whenever $\Delta \vdash s \to_\theory{R/E}^* t$ and $\Delta \vdash s \to_\theory{R/E}^* u$,  there exist terms $t',u'$ such that $\Delta \vdash t \to_\theory{R/E}^* t'$, $\Delta \vdash u \to_\theory{R/E}^* u'$ and $\Delta\vdash t'\ealeq u'$.
Also, $\theory{R}$ is said to be \emph{$\theory{E}$-terminating} if there is no infinite $\to_\theory{R/E}$ sequence.
 $\theory{R}$ is called \emph{$\theory{E}$-convergent} if it is $\theory{E}$-confluent and $\theory{E}$-terminating.
\end{definition}

The following example illustrates an ENRS to compute the prenex normal form of a first-order formula. We consider the commutativity of the connectives $\wedge$ and $\vee$.
\begin{example}[Prenex normal form rules]\label{ex:prenex-rules}
    Consider the signature for the first-order logic $\Sigma = \{\forall, \exists, \lnot, \land, \lor \}$, let $\theory{C} = \{\;\vdash P \lor Q \approx Q \lor P, \;\vdash P \land Q \approx Q \land P\}$ be the commutative theory. The prenex normal form rules can be specified by the following set $\theory{R}$ of nominal rewrite rules:
{\small
\[
\begin{array}{lrlc@{\hspace{.2cm}}lrl}
    R_1: & a \# P & \vdash \; P \land \forall[a]Q \rightarrow \forall [a](P \land Q)& \ &R_5: & &\vdash \; \lnot (\exists[a]Q) \rightarrow \forall[a] \lnot Q\\
    R_2: & a \# P & \vdash \; P \lor \forall[a] Q \rightarrow \forall[a] (P \lor Q)&&R_6: && \vdash \; \lnot (\forall[a]Q) \rightarrow \exists[a] \lnot Q\\
    R_3: & a \# P & \vdash \; P \land \exists[a]Q \rightarrow \exists[a](P \land Q)&&R_7: &&\vdash \; \exists[a](\forall[b]Q) \rightarrow \forall[b](\exists[a]Q)\\
    R_4: & a \# P & \vdash \; P \lor \exists[a]Q \rightarrow \exists[a](P \lor Q)&&&&\\
    % R_5: & & \vdash \; \lnot (\exists[a]Q) \rightarrow \forall[a] \lnot Q&&&&\\
    % R_6: & & \vdash \; \lnot (\forall[a]Q) \rightarrow \exists[a] \lnot Q&&&&\\
    % R_7: & & \vdash \; \exists[a](\forall[b]Q) \rightarrow \forall[b](\exists[a]Q)&&&&\\
\end{array}
\]
}
\end{example}

In Definition~\ref{def:RE-rewriting}, the relation $\to_\theory{R/E}$ operates on $\alpha,\theory{E}$-congruence classes of terms. These classes may be infinite—even when $\theory{E}$ itself induces only finite congruence classes—due to the presence of abstracted names and the infinite possibilities for $\alpha$-renaming, as illustrated in Examples~\ref{ex:overview1} and~\ref{ex:overview2}.

While the pure $\alpha$-equivalence relation $\approx_\alpha$ is decidable, combining it with an equational theory $\theory{E}$ that admits infinite congruence classes can lead to undecidability of the relation $\to_{\theory{R/E}}$, similarly to what happens in the first-order setting. To address this, the nominal rewriting relation $\to_{\theory{R,E}}$ uses nominal $\theory{E}$-matching, avoiding the need to inspect the entire $\alpha,\theory{E}$-congruence class of a term.

\begin{definition}[Nominal \theory{R,E}-rewriting]\label{def:rew-modC}
The \emph{one-step \theory{R,E}-rewrite relation} $\Delta\vdash s\rightarrow_{\theory{R,E}} t$ is the least relation such that there exists $R = (\nabla\vdash l\rightarrow r)\in \theory{R}$, position $\context{C}$, term $s'$,
% permutation $\pi$,
and substitution $\theta$,
    \begin{prooftree}
    \AxiomC{$s\equiv \context{C}[s']$}
    \AxiomC{$\Delta\vdash \big( \nabla\theta,\ s' \ealeq  l\theta, \ \context{C}[r\theta] \aleq t\big)$}
    \BinaryInfC{$\Delta\vdash s\rightarrow_{\theory{R,E}}\  t $}
    \end{prooftree}

The \theory{E}-\emph{rewrite relation} $\Delta \vdash s \rightarrow_{\theory{R,E}}^* t$ is
%is the reflexive transitive closure of $\rightarrow_{\theory{R,E}}$ relation, that is,
the least relation that includes $\rightarrow_{\theory{R,E}}$ and is closed by reflexivity and transitivity of $\to_{\theory{R,E}}$.
% , i.e., it satisfies:
% \begin{enumerate}
%     \item  for all $\Delta, s, s'$ we have $\Delta \vdash s \rightarrow_{\theory{R,E}}^* s'$ if $\Delta \vdash s \aleq s'$;
%     \item  for all $\Delta, s, t, u$,  $\Delta \vdash s \rightarrow_{\theory{R,E}}^* t$ and $\Delta \vdash t \rightarrow_{\theory{R,E}}^* u$ implies $\Delta \vdash s \rightarrow_{\theory{R,E}}^* u$.
% \end{enumerate}
If
$\Delta \vdash s \to_\theory{R,E}^* t$ and $\Delta \vdash s \to_\theory{R,E}^* u$, then we say that \emph{\theory{R,E} is \theory{E}-confluent} when there exist terms $t',u'$ such that $\Delta \vdash t \to_\theory{R,E}^* t'$, $\Delta \vdash u \to_\theory{R,E}^* u'$ and $\Delta\vdash t'\ealeq u'$.
\end{definition}
 A term-in-context $\Delta\vdash t$ is said to be in {\em \theory{R,E}-normal} form (\theory{R/E}-normal form) whenever one cannot apply another step of $\to_\theory{R,E}$ ($\to_\theory{R/E}$). Note that if the term $t$ is ground, the context is $\Delta=\emptyset$.

Note that $\rew \subseteq \erew\subseteq \ebarrew$.

\begin{example}[Cont. Example~\ref{ex:prenex-rules}]
    Consider the term $ s= S'\lor (\exists [a]Q' \lor P')$. Note that  to \theory{R,C}-reduce the term  $s$  it is necessary to \theory{C}-match the subterm $s' = \exists [a]Q' \lor P'$ with the lhs of the rule ($R_4$), getting the
    % permutation $\pi=id$ and
    substitution $\theta = [P \mapsto P', Q \mapsto Q']$.  Due to the freshness context of the rule $(R_4)$, we can only make the reduction under the context $\Delta = \{a\#P'\}$.
    %with the rule $a\#P \vdash P\lor \exists[a]Q \to \exists[a](P\lor Q)$.
    We need  to check that the following hold:
        (i) $ a\#P' \vdash a\#P' = (a\#P)\theta$;
        (ii)  $a\# P'\vdash \exists [a]Q' \lor P' \caleq (P \lor \exists [a]Q)\theta$; and
        (iii) $ a\# P'\vdash  \context{C}[(\exists [a](P\lor Q))\theta] =  S' \lor (\exists [a](P'\lor Q'))$. Thus the  one-step \theory{C}-rewrite is:
        $$a\#P' \vdash S'\lor (\exists [a]Q' \lor P') \to_{\theory{R,C}} S'\lor (\exists [a](P'\lor Q')).$$
    %Thus, $a\#P' \vdash S'\lor (\exists [a]Q' \lor P') \to_{\theory{R,C}} S'\lor (\exists [a](P'\lor Q'))$.

    Since $\lor$ is a commutative symbol, we could reduce the initial term to three other possible terms because we have two occurrences of the disjunction. Thus, we can ``permute'' the subterms inside the rewriting modulo~$\theory{C}$.
\end{example}

Next, we present an example that features an equational theory $\theory{\forall_C}$ that expresses the property that nested $\forall$-quantification  commute, note that this theory is not expressible in a first-order language:
\begin{example}\label{exa:rew-modulo}
Consider the rule $\theory{R}=\{c'\# X\vdash \forall [c']X\to X\}$   and the equational theory   $\theory{\forall_C}=\{\emptyset \vdash \forall[a]\forall [b] Y \approx \forall [b]\forall [a]Y\}$. Let  $s= \forall [a]\forall [d]R(a,Y)$, where $R$ is a binary relation in our signature.
    % \[
    %     [s]_{\alpha,\forall_C} = \left\{
    % \begin{array}{l}
    %   \forall[a](\forall[b](R(a,c))),  \forall[b](\forall[a](R(a,c))), \\
    %     %  \forall[b](\forall[b](R(b,c))), \forall[a](\forall[a](R(a,c))),\\
    %  \forall[d](\forall[b](R(d,c))), \forall[b](\forall[d](R(d,c))),\\
    %       \forall[a](\forall[d](R(a,c))), \forall[d](\forall[a](R(a,c))), \ldots\\
    % \end{array}
    % \right\}
    % \]
%
Note that if $\Delta=\{a,b,c',d\#Y\}$ we have $\Delta \vdash \forall[a]\forall[d]R(a,Y) \to_{\theory{R/\forall_C}} \forall[b]R(b,Y)$ since:
$$\Delta\vdash \forall[a]\forall[d]R(a,Y)\approx_{\theory{\forall_C}}\forall[d]\forall[a]R(a,Y)\to_{\theory{R}} \forall[a]R(a,Y)\approx_{\theory{\forall_C}}\forall[b]R(b,Y)$$
and no other rewriting step  with $\to_\theory{R}$ is possible; thus, $\Delta\vdash \forall [b]R(b,Y)$ is a normal form in the class $s\downarrow_\theory{R/\forall_C}$.

Observe that, at position $\context{C}=\forall [a][\_]$ of $s$ we have $s\equiv \context{C}[\forall [d]R(a,Y)]$
and we need to solve the matching problem  $\forall[d]R(a,Y)\approx_? \forall[c']X$ whose solution is $\theta=[X\mapsto R(a,Y)]$.
Thus, there is a step of rewriting $\to_{\theory{R,\forall_C}}$:

$$\Delta \vdash s= \forall[a]\forall[d]R(a,Y) \to_{\theory{R},\forall_\theory{C}} \forall[a]R(a,Y)$$
and no other $\to_\theory{R,\forall_C}$-step is possible. Thus,  $\Delta\vdash \forall[a]R(a,Y)$ is a normal form in the class  $s\downarrow_\theory{R,\forall_C}$.
Note that $\Delta\vdash s\downarrow_\theory{R,\forall_C}\approx_{\alpha/{\forall_\theory{C}}} s\downarrow _\theory{R/\forall_C}$.
\end{example}

Note that this behavior becomes more intricate because the term $s$ contains meta-variables, as we must then account for the associated freshness context $\Delta$—an issue already noted in Example~\ref{ex:overview1} in the
Introduction. In particular, applying rewrite steps and reasoning modulo
$\theory{E}$
 requires careful management of both binding structures and freshness conditions

 Following a similar approach to the classical setting of first-order rewriting modulo equational theories developed by~\citep{JouannaudKK83:Incremental}, we establish a connection between the nominal relations $\to_{\theory{R/E}}$ and $\to_{\theory{R,E}}$, through  an extension of the notion of {\em $\theory{E}$-coherence}, here formulated under a freshness context $\Delta$. This property ensures that rewriting steps respect $\alpha$-equivalence modulo
 \theory{E}  in a controlled way under a freshness context.
\begin{definition}[\(\theory{E}\)-Coherence] \label{def:e-coherence}
The relation $\to_{\theory{R,E}}$ is called \emph{$\theory{E}$-coherent} iff for all $ \Delta, t_1, t_2, t_3$ such that $\Delta\vdash t_1\ealeq t_2$ and $\Delta\vdash t_1 \to_{\theory{R,E}} t_3$, there exist $t_4, t_5, t_6$ such that $\Delta\vdash t_3 \to_{\theory{R,E}}^* t_4$, $t_2 \to_{\theory{R,E}} t_5 \to_{\theory{R,E}}^* t_6$ and $\Delta\vdash t_4 \ealeq t_6$, for some $\Delta$.
\end{definition}

\begin{theorem}%[Compatibility]
%[\(\theory{E}\)-Coherence]
\label{prop:nom-e-coherence}
Let
% \theory{E} be a first-order theory,
% and
\theory{R} be a nominal rewrite system that is
% \theory{E}-confluent and
\theory{E}-terminating
and \theory{R,E} be \theory{E}-confluent.
Then the \theory{R,E}- and \theory{R/E}-normal forms of any term-in-context $\Delta\vdash t$ are \theory{E}-equal iff $\to_{\theory{R,E}}$ is \theory{E}-coherent.
\end{theorem}
\begin{proof}
    % By case analysis and analysis of the \theory{R/E}- and \theory{R,E}-normal forms.

    ($\Rightarrow$) Suppose $\Delta\vdash t{\downarrow}_\theory{R/E} \ealeq t{\downarrow}_\theory{R,E}$. %$\Delta\vdash t$.
 Let $t_2$ and $t_3$ be terms such that $\Delta \vdash t \ealeq t_2$ and $\Delta \vdash t \erew t_3 \erew^* t_3{\downarrow_\theory{R,E}}$. By hypothesis, $\Delta\vdash t_3{\downarrow}_\theory{R/E} \ealeq t_3{\downarrow}_\theory{R,E}$, thus $t_3{\downarrow}_\theory{R,E}$ is also a \theory{R/E}-normal form of $t_2$ under $\Delta$, then $\Delta\vdash t_2{\downarrow}_\theory{R/E} \ealeq t_3{\downarrow}_\theory{R/E}$. Again, using the hypothesis, we have $\Delta\vdash t_2{\downarrow}_\theory{R,E} \ealeq t_2{\downarrow}_\theory{R/E}$, and by transitivity of $\ealeq$, we get $\Delta\vdash t_2{\downarrow}_\theory{R,E} \ealeq t_3{\downarrow}_\theory{R,E}$.
    Furthermore, since \theory{R} is \theory{E}-terminating, $t_2{\downarrow}_\theory{R,E}$ cannot be
     equal to $t_2$, otherwise we would have a term $t_1$ such that $\Delta\vdash t_2 \ealeq t_1 \erew^{+} t_3{\downarrow}_\theory{R,E} \ealeq t_2$, which is $\Delta\vdash t_2 \ebarrew^{+} t_2$ a cycle of \theory{R/E}.
    Therefore $\erew$ is \theory{E}-coherent.

    ($\Leftarrow$) Suppose  that $\erew $ is \theory{E}-coherent. We prove that $\Delta \vdash t{\downarrow}_\theory{R,E} \ealeq t{\downarrow}_\theory{R/E}$ for any term $t$, by induction on \theory{R/E}.
    Let $t$ be arbitrary and such that $t \equiv t{\downarrow}_\theory{R/E}$, then it is also in \theory{R,E}-normal form and we are done.
    Consider then that $\Delta\vdash t \ebarrew^+ t{\downarrow}_\theory{R/E}$, that is, $\Delta \vdash t\ealeq t_1 \rew t_2 \ebarrew^* t{\downarrow}_\theory{R/E}$. Observe that $t_1 \not\equiv t_1{\downarrow}_\theory{R,E}$ because we have $\Delta\vdash t_1 \rew t_2$ and $\theory{R} \subseteq \theory{R,E}$, thus $\Delta\vdash t_1 \erew^+ t_1{\downarrow}_\theory{R,E}$.

    By \theory{E}-coherence of $\erew$, we have $\Delta\vdash t \erew^+ t{\downarrow}_\theory{R,E}$ and $\Delta\vdash t{\downarrow}_\theory{R,E} \ealeq t_1{\downarrow}_\theory{R,E}$.
    Now, we apply the \theory{E}-confluence property of \theory{R,E} on $\Delta \vdash t_1 \erew^+ t_1{\downarrow}_\theory{R,E}$ and $\Delta\vdash t_1 \rew t_2$ (hence $\Delta\vdash t_1 \erew t_2$) and we have $\Delta \vdash t_2 \erew^* t_2{\downarrow}_\theory{R,E}$ and $\Delta\vdash t_1{\downarrow}_\theory{R,E} \ealeq t_2{\downarrow}_\theory{R,E}$. Finally we can apply the induction on $t_2$, since $\Delta\vdash t \ebarrew t_2$, then $\Delta\vdash t_2{\downarrow}_\theory{R/E} \ealeq t_2{\downarrow}_\theory{R,E}$ and the result follows by transitivity of $\ealeq$.

% \begin{figure}[!t]
% \begin{center}
% {\small
% \boxed{
% $$
% \xymatrix{
% t\ar@{-->}[dd]^{+}_{\teal{\theory{R,E}}} \ar@{~}[rr]
% _{\ealeq} && t_1\ar[dd]^{+}_{\theory{R,E}}  \ar[rr]_{\theory{R}} && t_2\ar@{-->}[dd]^{*}_{\teal{\theory{R,E}}}   \ar[rr]^{*}_{\theory{R/E}} && t{\downarrow}_\theory{R/E}\\
% & \circled{1} && \circled{2} & \ar@{}[ur]_{\circled{3}} &&\\
% t{\downarrow} \ar@{~~}[rr]_{\teal{\ealeq}}  && t_1{\downarrow}\ar@{~~}[rr]_{\teal{\ealeq}} && t_2{\downarrow} \ar@{~~}[uurr]_{\teal{\ealeq}} &&
% }
% $$
% }}
% \end{center}
% \caption{\theory{E}-coherence diagram \alert{You need a $\Delta$}}
% \end{figure}
\end{proof}

Note that  $\theory{R,E}$-reducibility is decidable if  $\theory{E}$-matching is decidable. The existence of a finite and complete $\theory{E}$-unification
algorithm is a sufficient condition for that decidability~\citep{JouannaudKK83:Incremental}.
However, solving nominal $\theory{E}$-unification problems has the additional complication of dealing with $\alpha$-equality, significantly impacting obtaining finite and complete sets of nominal $\theory{E}$-unifiers.

\begin{example}\label{rmk:cunif-notfin}
Nominal $\theory{C}$-unification is not finitary when solutions are represented using freshness constraints and substitutions~\citep{FormalisingNomC-unif/mscs/Gabriel21}, but the type of problems that generate an infinite set of $\theory{C}$-unifiers are fixed-point equations $\pi\cdot X \; _?{\overset{\theory{C}}{\approx}}_? \; X$. E.g., the \theory{C}-unification problem $(a \ b)\cdot X\; _?{\overset{\theory{C}}{\approx}}_? \; X$  has solutions $[X\mapsto a\oplus b], [X\mapsto (a\oplus b)\oplus (a \oplus b)],\ldots$ (Example~ \ref{ex:c-unification}). However, these problems do not appear in nominal $\theory{C}$-matching, which is finitary~\citep{FormalisingNomC-unif/mscs/Gabriel21}. Thus, the relation $\to_{\theory{R,C}}$ is decidable. %We will come back to this issue later on.
\end{example}

\subsection{Closed Rewriting Modulo Theories}

As briefly discussed in Section~\ref{ex:overview1}, for efficiency, we will focus on working with closed systems. Intuitively, closed terms contain no free atoms, and closed axioms are identities between closed terms that prevent any abstracted atoms from becoming free. Before formally introducing the notions of closed terms and closed systems, we first define the concept of a {\em freshened variant} of a term or a context: given a term $t$, we say that $\nw{t}$ is a \emph{freshened variant} of $t$ when $\nw{t}$ has the same structure of $t$, except that the atoms and unknowns have been replaced by `fresh' ones.
    Similarly, if $\nabla$ is a freshness context then $\nw{\nabla}$ will denote a freshened variant of $\nabla$, that is, if $a\# X \in \nabla$ then $\nw{a}\# \nw{X} \in \nw{\nabla}$ where $\nw{a}$ and $\nw{X}$ are chosen fresh. This notion naturally extends to other syntactic objects such as equality and rewrite judgements.

    \begin{example}
    We have that $[\nw{a}][\nw{b}]\nw{X}$ is a freshened variant of $[a][b]X$. Also $\nw{a}\#\nw{X}$ is a freshened variant of $a\# X$, and $\emptyset \vdash f([\nw{a}]\nw{X}) \to [\nw{a}]\nw{X}$ is a freshened variant of $\emptyset \vdash f([a]X) \to [a]X$.
    Note that neither $[\nw{a}][\nw{a}]\nw{X}$ nor $[\nw{a}][\nw{b}]X$ are freshened variants of $[a][b]X$: the first one because we are identifying different atoms with the same fresh name, and the second one because we did not freshen $X$.
    \end{example}

    \begin{definition}[Closedness]
    A term-in-context $\nabla \vdash l$ is \emph{closed} if there exists a solution for the matching problem  $(\nw{\nabla} \vdash \nw{l})\; _?{\approx}\; (\nabla, A(\nw{l})\# V(\nabla, l) \vdash l).$
    A rule $R=(\nabla \vdash l \to r)$ and an axiom $Ax = (\nabla \vdash l = r)$ are called \emph{closed} when $\nabla \vdash (l,r)$ is closed.

    \end{definition}

    \begin{example}\label{ex:lambda}
    Consider the following rewrite rules \theory{R_\texttt{lam}} for the $\lambda$-calculus:
    {\normalfont
    {\small
    $$
    \begin{array}{rclcl}
      (\beta)\qquad \qquad   & \vdash & \texttt{app}(\texttt{lam}([a]X),X') & \to & \texttt{sub}([a]X,X')   \\
        & \vdash & \texttt{sub}([a]a,X) & \to & X \\
        a\# Y & \vdash & \texttt{sub}([a]Y,X) & \to & Y \\
        & \vdash & \texttt{sub}([a]\texttt{app}(X,X'),Y) & \to & \texttt{app}(\texttt{sub}([a]X,Y),\texttt{sub}([a]X',Y)) \\
        b\# Y & \vdash & \texttt{sub}([a]\texttt{lam}([b]X),Y) & \to & \texttt{lam}([b]\texttt{sub}([a]X,Y)) \\
    \end{array}
    $$
    }}

    \noindent All the rewrite rules above are closed.
    We refer to \theory{R_\texttt{sub}} as the rules of \theory{R_\texttt{lam}} without rule $(\beta)$.
    \end{example}

    \begin{example}
    Consider the rule $R \equiv \emptyset \vdash [a]f(a, X) \to a$. This rule is not closed because there is no solution to
     $(\emptyset \vdash  [a']f(a', X'), a') \; _?{\approx} \; (a'\#X \vdash [a]f(a, X), a)$.
Similarly, the rule     $R \equiv \emptyset \vdash [a]f(a, X) \to X$ is not closed--the abstracted atom $a$ in the left-hand side can become free when applying the rule (e.g., if $X$ is instantiated with $a$). In Example~\ref{ex:lambda}, the third rule has a freshness constraint $a\# Y$, which prevents this.
    \end{example}

\begin{definition}
    [Closed Nominal \theory{R,E}-rewriting]\label{def:clsd-rew-modC}
    The \emph{one-step closed \theory{R,E}-rewrite relation} $\Delta\vdash s\rightarrow_{\theory{R,E}}^{c} t$ is the least relation generated by the rule below, where $R = (\nabla\vdash l\rightarrow r)\in \theory{R}$,  $\Delta \vdash s$ is a term-in-context, and $\nw{R}$ a freshened variant of $R$ (so fresh for $R, \Delta, s, t$),  $\context{C}$ a position,
    %permutation $\pi$,
    and $\theta$ a substitution such that
        \begin{mathpar}
        \inferrule{s\equiv \context{C}[s'] \\ \Delta, A(\nw{R})\# V(\Delta, s,t)\vdash \big( \nw{\nabla}\theta,\ s' \ealeq \nw{l}\theta, \ \context{C}[\nw{r}\theta] \aleq t\big)}
        {\Delta\vdash s\rightarrow_{\theory{R,E}}^{c}\  t }
        \end{mathpar}

    \end{definition}

\section{Nominal Narrowing Modulo Theories}\label{sec:e-narrowing}

In this section,  we define the nominal narrowing relation induced by an equivariant set of nominal rules \theory{R} modulo an equational theory $\theory{E}$, extending the framework of~\citep{NominalNarrowing16}. Nominal narrowing generalises nominal rewriting as it considers nominal unification instead of matching modulo \theory{E}. As expected,  this process must carefully handle the freshness constraints associated with both the rules in \theory{R} and the identities in the theory \theory{E}, while operating modulo renaming of bound names. The main goal of this section is to establish the nominal version of the {\em Lifting Theorem Modulo \theory{E}} (Theorem~\ref{teo:Elifting}), a fundamental result that connects \theory{R,E}-narrowing steps to \theory{R,E}-rewriting steps. This theorem is crucial for ensuring that the narrowing relation can be effectively used to solve \theory{(R\cup E_\alpha)}-unification problems.

\begin{definition}[Nominal \theory{R,E}-narrowing]\label{def:narr-C}
The \emph{one-step \theory{R,E}-narrowing relation} $(\Delta\vdash s) \rightsquigarrow_{\theory{R,E}} (\Delta' \vdash t)$ is the least relation such that for some $(\nabla\vdash l\rightarrow r)\in \theory{R}$,  position $\context{C}$, term $s'$,
% permutation $\pi$,
and substitution $\theta$, we have
    % \vspace{-2mm}
    \begin{prooftree}
    \AxiomC{$s\equiv \context{C}[s']$}
    \AxiomC{$\Delta'\vdash \big(\nabla\theta,~\Delta\theta,~ s'\theta \ealeq l\theta, \,(\context{C}[r])\theta \aleq t\big)$}
    \RightLabel{.}
    \BinaryInfC{$(\Delta\vdash s) \rightsquigarrow_{\theory{R,E}}^\theta (\Delta' \vdash t) $}
    \end{prooftree}
where $(\Delta',\theta) \in {\cal U}_{\theory{E}}( \nabla\vdash l,  \Delta\vdash  s')$.  We will write only $(\Delta\vdash s) \rightsquigarrow_{\theory{R,E}} (\Delta' \vdash t)$, omitting the $\theta$, when it is clear from the context.
The \emph{nominal \theory{R,E}-narrowing relation} $(\Delta \vdash s) \rightsquigarrow_{\theory{R,E}}^* (\Delta' \vdash t)$ is
the least relation that includes $\rightsquigarrow_{\theory{R,E}}$ and is closed by reflexivity and transitivity of $\rightsquigarrow_{\theory{R,E}}$.
\end{definition}
The
% permutation $\pi$ and
substitution $\theta$ in the definition above
% are
is
found by solving the nominal  $\theory{E}$-unification problem $(\nabla\vdash l) \; _?{\overset{\theory{E}}{\approx}}_? \; (\Delta\vdash  s')$.

\begin{remark}
%Note that
The decidability of $\rightsquigarrow_\theory{R,E}$ relies on the existence of an algorithm for nominal \theory{E}-unification.
% that generates a finite minimal set of solutions.
%In this work,
In our examples we will use the theory \theory{C},
 which has a nominal unification algorithm.
\end{remark}

Since nominal $\theory{C}$-narrowing uses nominal $\theory{C}$-unification, which is not finitary when we use pairs $(\Delta', \theta)$ of freshness contexts and substitutions to represent solutions, following Example~\ref{rmk:cunif-notfin}, we conclude that our nominal $\theory{C}$-narrowing trees are infinitely branching. The following example illustrates these infinite branches.
% \begin{figure}[!t] %{r}{0.5\textwidth}
%     \centering
%     \resizebox{10cm}{!}{
%     \includegraphics[scale=.69]{Figs/narrowtree5 - Copia.png}
%     }
%     \caption{Infinitely branching tree}
%     \label{fig:enter-label5}
%     \end{figure}

% \begin{figure}[!t] %{r}{0.5\textwidth}
%  %   \centering
%     %\begin{mdframed}
%     \resizebox{\textwidth}{!}{
%     \includegraphics{Figs/narrowtree8.png}
%     }
%     %\end{mdframed}
%     \caption{Infinitely branching tree}
%     \label{fig:enter-label5}
%     \end{figure}

\begin{figure}[!t]
$$
\xymatrix{
& \emptyset \vdash h(f^\theory{C}([b][a]X,X)) \ar@{~>}[d]_{\theta_0} & &\\
& \emptyset\vdash f^\theory{C}([b][a]X,X) \ar@{~>}[dl]^{\theta_1} \ar@{~>}[d]_{\theta_2} \ar@{~>}[dr]_{\theta_n} \ar@{~>}[drr] &  &\\
a\#X,b\#X \vdash t_1 & \emptyset\vdash t_2 \ar@{}[r]^{\ldots} & \qquad \emptyset\vdash t_n \; \ar@{}[r]^{\ldots}& \qquad\qquad\\
}
$$

\footnotesize{
$\theta_0=[Y\mapsto f^\theory{C}([b][a]X,X)]$

$\theta_1=[Z\mapsto X]$ \hfill $t_1 =f^\theory{C}(h(X),h(X))$

$\theta_2=[Z\mapsto X][X\mapsto a\oplus b]$ \hfill $t_2 = f^\theory{C}(h(a\oplus b),h(a\oplus b))$

$\theta_n=[Z\mapsto X][X\mapsto (a\oplus b)\oplus(b\oplus a)]
$ \hfill $t_n = f^\theory{C}(h((a\oplus b)\oplus(b\oplus a)),h((a\oplus b)\oplus(b\oplus a)))$}
\caption{Infinitely branching tree}\label{fig:enter-label5}
\end{figure}

\begin{example}[Cont. Example~\ref{ex:c-unification}]\label{ex:narrow}
    Consider the signature $\Sigma = \{h:1, f^\theory{C}:2, \oplus: 2\}$, where $f^\theory{C}$ and $\oplus$ are commutative symbols.
    Let $R=\{ \ \vdash h(Y) \to Y ,\ \vdash f^\theory{C}([a][b]\cdot Z, Z) \to f^\theory{C}(h(Z),h(Z)) \}$ be a set of rewrite\footnote{$\vdash l\to r$ denotes $\emptyset\vdash l\to r$.} rules.
    Let $\ \vdash h(f^\theory{C}([b][a]X,X))$ be a nominal term that we want to apply nominal \theory{C}-narrowing to.
Observe that we can apply one step of narrowing, and then we obtain a branch that yields infinite branches due to the fixed-point equation (see Fig.~\ref{fig:enter-label5}).
    The first narrowing step is $$\emptyset \vdash h(f^\theory{C}([b][a]X,X)) \enarrow{R,C} \emptyset \vdash f^\theory{C}([b][a]X,X),$$ using the rule $\ \vdash h(Y) \to Y$.
    The substitution $\theta_0= [Y\mapsto f^\theory{C}([b][a]X,X)]$ as well as the other (infinite) narrowing steps were computed in Example~\ref{ex:c-unification}.

\end{example}

The following proposition shows that each nominal narrowing step corresponds to a nominal rewriting step, using the same substitution $\theta$.

\begin{proposition}\label{prop:narrow-to-rewriting}
Let \theory{E} be an equational theory for which a complete \theory{E}-unification algorithm exists. $(\Delta_0\vdash s_0) \rightsquigarrow_{\theory{R,E}}^{\theta} (\Delta_1 \vdash s_1) $ implies $\Delta_1\vdash (s_0 \theta ) \rightarrow_{\theory{R,E}}  s_1 $.
\end{proposition}

\begin{proof}
    Indeed, suppose we have $(\Delta_0\vdash s_0) \rightsquigarrow_{\theory{R,E}}^{\theta} (\Delta_1 \vdash s_1) $. The narrowing step guarantees that for a substitution $\theta$
    % , some permutation $\pi$,
    and a rule $\nabla\vdash l\to r\in \theory{R}$, the following holds: $s_0 \equiv \context{C}[s'_0]$ and $\Delta_1\vdash \big(\nabla\theta,~\Delta_0\theta,~ s'_0\theta \ealeq
    % \pi\cdot
    l\theta, \,(\context{C}[
    % \pi\cdot
    r])\theta \aleq s_1\big)$.
    It it clear that $s_0\theta \equiv \context{C}\theta[s_0'\theta]$, and by the definition of \theory{R,E}-rewrite, the previous items give us $\Delta_1\vdash s_0\theta\erew s_1$.
\end{proof}

\subsection{Nominal Lifting Theorem modulo \texorpdfstring{\theory{E}}{E}}\label{section:NomLiftingModuloC}

We now extend the previous result to derivations.

\begin{remark}Throughout this section, we assume that $\theory{R{\cup}E_\alpha}$ is {\em well-structured}, that is,   $\theory{R}=\{\nabla_i\vdash l_i\to r_i\}$ is $\theory{E}$-convergent
and  there exists a complete \theory{E}-unification algorithm.
\end{remark}
%R,E well-structured ENRS

Before stating the main result, we introduce the  notion of a normalised substitution with respect to the relation $\to_{\theory{R,E}}$.
\begin{definition}[Normalised substitution w.r.t $\to_{\theory{R},\theory{E}}$]
A substitution $\theta$ is \emph{normalised in $\Delta$ with respect to $\to_{\theory{R,E}}$} if $\Delta \vdash X\theta$ is a $\theory{R,E}$-normal form  for every $X$. A substitution $\theta$ {\em satisfies the freshness context} $\Delta$ iff there exists a freshness context $\nabla$ such that $\nabla \vdash a\#X\theta$ for each $a\#X \in \Delta$. In this case, we say that $\theta$ satisfies $\Delta$ with $\nabla$.
In the case where \theory{E=\emptyset} we write only $\rew$.
\end{definition}

The following result (correctness) states that a finite sequence of rewriting steps exists for each finite sequence of narrowing steps.

\begin{lemma}{($\rightsquigarrow_{\theory{R,E}}^*$ to $\to_{\theory{R,E}}^*$)}\label{theo:narrowtorewrite}
Let $\theory{R{\cup}E_\alpha}$ be well-structured. Consider the $(\Delta_0\vdash s_0) \rightsquigarrow^*_{\theory{R,E}} (\Delta_n \vdash s_n) $  nominal \theory{R,E}-narrowing derivation.
Let $\rho$ be a substitution satisfying $\Delta_n$ with  $\Delta$.
$$(\Delta_0\vdash s_0) \rightsquigarrow^{\theta_0}_{\theory{R,E}} (\Delta_1 \vdash s_1) \rightsquigarrow^{\theta_1}_{\theory{R,E}}  \ldots \rightsquigarrow^{\theta_{n-1}}_{\theory{R,E}} (\Delta_n \vdash s_n) $$
%such that $\Delta \vdash \Delta_0\rho$.
Then, there exists a nominal \theory{R,E}-rewriting derivation
% \vspace{-2mm}
$$\Delta\vdash s_0\rho_0  \rightarrow_{\theory{R,E}} \ldots \rightarrow_{\theory{R,E}} s_i\rho_i \rightarrow_{\theory{R,E}}\ldots   \rightarrow_{\theory{R,E}} s_{n-1}\rho_{n-1}\rightarrow_{\theory{R,E}} s_n\rho$$
such that $\Delta\vdash \Delta_i\rho_{i}$ and $\rho_{i}=\theta_{i}\ldots \theta_{n-1}\rho$, for all $0\leq i < n$. In other words,
% \vspace{-1mm}
$\Delta\vdash (s_0 \theta )\rho \rightarrow^*_{\theory{R,E}}  s_n \rho $
where $\theta=\theta_0\theta_1\ldots \theta_{n-1}$.
\end{lemma}
\begin{proof}
By induction on the length $n\geq 1$ of the narrowing derivation $(\Delta_0\vdash s_0) \rightsquigarrow^{n}_{\theory{R,E}} (\Delta_n \vdash s_n)$.
We start the induction for $n=1$ because the case for $n=0$ holds trivially and gives no additional insight.
\begin{itemize}
    \item \textbf{Base Case:} For $n=1$,
    % we have by assumption that $\rho_0 = \theta_0\rho$ and $\Delta \vdash \Delta_0\rho_0$. From Lemma~\ref{lem:correctness}, there is a step
    % $\Delta \vdash s_0\rho_0 = (s_0\theta_0)\rho \rightarrow_{\theory{R,C}} s_1\rho$, and the result follows.
    we have $(\Delta_0\vdash s_0) \rightsquigarrow_{\theory{R,E}}^{\theta} (\Delta_1 \vdash s_1)$, and it follows, from
    % From
    Proposition~\ref{prop:narrow-to-rewriting}, that
    % $(\Delta_0\vdash s_0) \rightsquigarrow_{\theory{R,E}}^{\theta} (\Delta_1 \vdash s_1) $ implies
    $\Delta_1\vdash (s_0 \theta ) \rightarrow_{\theory{R,E}}  s_1 $.
By \theory{E}-compatibility with substitutions, $\Delta_1 \vdash s_0\theta \rightarrow_{\theory{R,E}} s_1$ gives the following:

\begin{enumerate}[leftmargin=*]
    \item $(s_0\theta)\rho \equiv (\context{C}\theta[s'_0\theta])\rho = \context{C}\theta\rho[(s'_0\theta)\rho]$
    \item $\Delta_1 \vdash \nabla\theta$ implies  $\nfpair{\Delta_1\rho} \vdash \nabla\theta\rho$
    \item $\Delta_1 \vdash s_0'\theta \ealeq
    % \pi\cdot
    l\theta$ implies $\nfpair{\Delta_1\rho} \vdash s_0'\theta\rho \ealeq
    % \pi\cdot
    % l\theta\rho =
    % \pi\cdot
    l\theta\rho$
    \item $\Delta_1 \vdash \context{C}\theta[
    % \pi\cdot
    r\theta] \aleq s_1$ implies $\nfpair{\Delta_1\rho} \vdash \context{C}\theta\rho[
    % \pi\cdot
    r\theta\rho] = (\context{C}\theta[
    % \pi\cdot
    r\theta])\rho \aleq s_1\rho$
\end{enumerate}
which implies that $\Delta = \nfpair{\Delta_1\rho} \vdash (s_0\theta)\rho \rightarrow_{\theory{R,E}} s_1\rho$.
Note that we need $\rho$ satisfying $\Delta_1$ with $\Delta$ to guarantee that when we instantiate $\Delta_1$ we do not have any inconsistency with the freshness constraints in $\Delta_1$.

    Note that $\theta= \theta_0$ for the case $n=1$. The result above gives us
    % and by Lemma~\ref{lem:correctness}, for any $\rho$ satisfying $\Delta_1$ with $\Delta$ we have
    $\Delta \vdash  (s_0\theta_0)\rho \rightarrow_{\theory{R,E}} s_1\rho$. Since $\Delta\vdash \Delta_1\rho$, and by the narrowing step $\Delta_1\vdash \Delta_0\theta_0$, we get $\Delta\vdash \Delta_0\theta_0\rho$. Taking $\rho_0 = \theta_0\rho$, we have the result
    $\Delta\vdash s_0\rho_0 \to_\theory{R,E} s_1\rho$
    such that $\Delta\vdash \Delta_0\rho_0$.
    \item \textbf{Induction Step:} Assume that the result holds for $n>1$. Then
    $(\Delta_0 \vdash s_0) \rightsquigarrow^{n}_{\theory{R,E}} (\Delta_n \vdash s_n)$
    implies that there exists a rewriting derivation
    $\Delta \vdash s_0\rho_0 \rightarrow^{n}_{\theory{R,E}} s_n\rho$, for some $\rho$ satisfying $\Delta_n$ with $\Delta$
    and Figure~\ref{fig:narrrew} illustrates this setting.
    We will prove the result for $n+1$ steps.

    Consider the narrowing step
    $(\Delta_n \vdash s_n) \rightsquigarrow^{\theta_n}_\theory{R,E} (\Delta_{n+1} \vdash s_{n+1}).$
   By
   % Lemma~\ref{lem:correctness},
   the base case,
   for any substitution, say $\sigma$, that satisfies $\Delta_{n+1}$ with $\Delta$, that is (H1)~$\Delta\vdash \Delta_{n+1}\sigma$,
    % there exists $\Delta$ that satisfies $\Delta_i\rho_{i+1}$, and we have
    we have
    \begin{equation}\label{eq:narrow}
        \Delta \vdash (s_n\theta_n)\sigma \to_{\theory{R,E}} s_{n+1}\sigma
    \end{equation}

    %$\Delta \vdash s_0\rho_0 \rightarrow^{n+1}_{\theory{R,E}} s_{n+1}\sigma$, for some $\sigma$ that satisfies $\Delta_{n+1}$ with $\Delta$. That is, $$\Delta \vdash s_0\rho_0 \to^n_\theory{R,E} s_n\rho \to_{\theory{R,E}} s_{n+1}\sigma.$$

    Take $\rho = \theta_n\sigma$. Note that  $\rho$ satisfies $\Delta_n$ with $\Delta$:
    since by Definition~\ref{def:narr-C}, we have (H2)~$\Delta_{n+1}\vdash \Delta_n\theta_n$  and this with Proposition~\ref{def:e-compatible}(1) give us  (H3)~$\langle\Delta_{n+1}\sigma\rangle_{nf} \vdash \Delta_n\rho$ .
    % \begin{enumerate}[\bf (H1)]
    %     % \item $\Delta\vdash \Delta_{n+1}\sigma$.
    %     \item By Definition~\ref{def:narr-C}: $\Delta_{n+1}\vdash \Delta_n\theta_n$.
    %     \item From {\bf (H2)} and Proposition~\ref{def:e-compatible}(1):
    %     % generalised to \theory{E}:
    %     $\langle\Delta_{n+1}\sigma\rangle_{nf} \vdash \Delta_n\rho$.
    % \end{enumerate}
     Thus, from  (H1) and (H3) it follows that  $\Delta\vdash \Delta_n\rho$.
    By the induction hypothesis,  we have
    $\Delta \vdash s_0\theta_0\ldots\theta_{n-1}\rho \rightarrow^{n}_{\theory{R,E}} s_n\rho$
    with $\Delta\vdash \Delta_i\rho_i$
    and $\rho_i = \theta_i\ldots\theta_{n-1}\rho$, for every $i=0,\ldots, n$. Hence,
    $$\Delta \vdash s_0\theta_0\ldots\theta_{n-1}\theta_n\sigma \to^{n}_{\theory{R,E}} s_n\theta_n\sigma \overset{\eqref{eq:narrow}}{\to}_{\theory{R,E}} s_{n+1}\sigma,$$
    and the result follows.
%
    % Note that $\Delta\vdash \Delta_{n+1}\sigma$ and by the narrowing step $\Delta_{n+1}\vdash \Delta_n\theta_n$ thus $\Delta\vdash \Delta_n\theta_n\sigma = \Delta_n\rho$.
\end{itemize}
%
%Notice that $\rho_i$ is $\theory{R,E}$-normalised if $\rho_0$ is $\theory{R,E}$-normalised.
\end{proof}

\begin{figure}[!t]
    \centering
    \includegraphics[width=\textwidth]{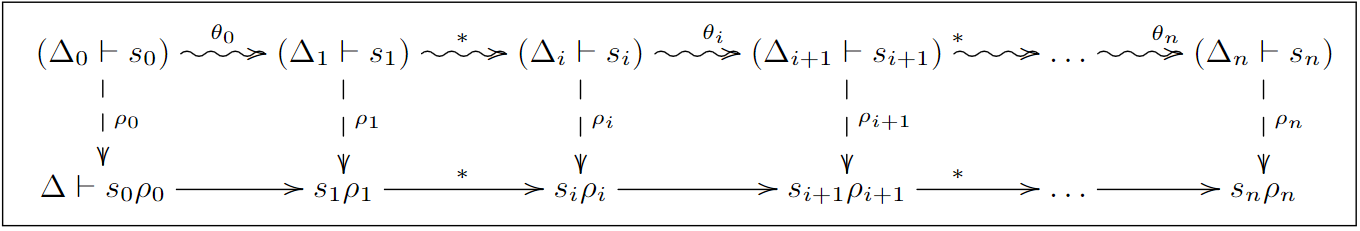}
    \caption{Corresponding Narrowing to Rewriting Derivations}
    \label{fig:narrrew}
\end{figure}

The proof of the converse (completeness) is more challenging.

\begin{lemma}{($\to^*_{\theory{R,E}}$ to $\rightsquigarrow^*_{\theory{R,E}}$)}\label{theo:rewritenarr}
Let $\theory{R{\cup}E_\alpha}$ be well-structured  and  $\to_{\theory{R,E}}$ be $\theory{E}$-coherent.
Let $V_0$ be a finite set of variables containing $V=V(\Delta_0,s_0)$. Then, for any $\theory{R,E}$-derivation
%---
$$\Delta \vdash t_0 = s_0\rho_0\rightarrow_{\theory{R,E}} t_1\rightarrow_{\theory{R,E}}\ldots \rightarrow_{\theory{R,E}} t_n=t_0{\downarrow}$$
%-----
to any of its $\theory{R,E}$-normal forms, say $t_0{\downarrow}$, where $\texttt{dom}(\rho_0) \subseteq V(s_0) \subseteq V_0$ and $\rho_0$ is a $\theory{R,E}$-normalised substitution that satisfies $\Delta_0$ with $\Delta$, there exist a $\theory{R,E}$-narrowing derivation
%------
$$(\Delta_0\vdash s_0) \rightsquigarrow^{\theta_0}_{\theory{R,E}} (\Delta_1 \vdash s_1) \rightsquigarrow^{\theta_1}_{\theory{R,E}}  \ldots \rightsquigarrow^{\theta_{n-1}}_{\theory{R,E}} (\Delta_n \vdash s_n) $$
%---
such that, for each $i=1,\ldots, n$,  $\rho_{i}=\theta_i\ldots \theta_{n-1}\rho_n$, $\theta=\theta_0\theta_1\ldots \theta_{n-1}$, and $V(s_i)\subseteq V_i$, the following hold:\\

\begin{minipage}{7cm}
\begin{enumerate}
\item $\Delta\vdash \Delta_i\rho_i$; \label{cond1}
\item \label{cond2} $\Delta \vdash s_i\rho_i \ealeq t_i$;
\end{enumerate}
\end{minipage}
\begin{minipage}{7cm}
\begin{enumerate}[start=3]
\item \label{cond3} $\Delta \vdash \rho_0|_V \ealeq \theta\rho_n|_V$;
\item \label{cond4} $\texttt{dom}(\rho_i) \subseteq V_i$.
\end{enumerate}
\end{minipage}
\end{lemma}
\begin{proof}
 By induction on the number of steps $n$ applied in the derivation $\Delta \vdash t_0 = s_0\rho_0\rightarrow_{\theory{R,E}}^* t_0{\downarrow}$.
 In fact, the base case does not rely on $t_1$ to be a normal form and we prove a more general base case.

\begin{itemize}
    \item \textbf{Base Case:}
    % For $n=1$ the result follows directly from Lemma~\ref{lem:rewrite-narr}.
    Let $n=1$.  The one-step rewriting is done in a position $\context{C}_0$ of $t_0$, with substitution $\sigma$ and  rule $R_0 = \nabla_0 \vdash l_0 \to r_0 \in \theory{R}$:
%$\Delta \vdash t_0 \to_{[\context{C}_0, R_0],\theory{E}} t_1$. Then,
\begin{mathpar}
    \inferrule*[left=(*)]{t_0\equiv \context{C}_0[t_0'] \\ \Delta \vdash \nabla_0\sigma,\; t_0' \ealeq
    % \pi\cdot
    l_0\sigma  ,\; \context{C}_0[
    % \pi\cdot
    r_0\sigma] \aleq t_1}{\Delta \vdash t_0 \to_{[\context{C}_0, R_0],\theory{E}} t_1}
\end{mathpar}

Note that the variables of $R_0$ are renamed w.r.t. $t_0 = s_0\rho_0$ and $\Delta$. Thus,  $V(R_0)\cap V(\Delta, t_0) = \emptyset$ and $\texttt{dom}(\sigma) \cap V_0 = \emptyset$. By hypothesis,  (H1)~$\Delta\vdash \Delta_0\rho_0$. Since $\rho_0$ is  normalised in $\Delta$ and $\Delta\vdash s_0\rho_0\to_\theory{R,E} t_1$,  there must exist a non-variable position $\context{C}_0'$ and a subterm $s_0'$ of $s_0$ such that $s_0\equiv \context{C}_0'[s_0']$ and (H2)~$\Delta \vdash s_0'\rho_0 \ealeq t_0' \ealeq
    % \pi\cdot
    l_0\sigma$.

% Define {\bf (H4)} $\theta=\rho_0\sigma$.
Consider the substitution $\eta = \rho_0\cup\sigma$. Then the pair $(\Delta, \eta)$ is a solution for $(\Delta_0\vdash s_0') \unif (\nabla_0\vdash
% \pi\cdot
l_0)$. That is, the following hold:
\begin{itemize}[leftmargin=*]
    \item $\Delta\vdash s_0'\eta \ealeq
    % \pi\cdot
    l_0\eta$: from (H2).
    \item $\Delta\vdash \Delta_0\eta$: from (H1) and the fact that $\sigma$ does not affect $\Delta_0$, because $\dom{\sigma}\cap V_0 = \emptyset$ and $V_0$ contains $V(\Delta_0)$.
    \item $\Delta\vdash \nabla_0\eta$: from (*) and the fact that $\rho_0$ does not affect $\nabla_0$, as $\dom{\rho_0}\cap V(R_0)= \emptyset$, $\ V(\nabla_0)\subseteq V(R_0))$ and $\dom{\rho_0} \subseteq V(s_0) \subseteq V(t_0)$.
\end{itemize}

Now, we can take a more general solution
% than $(\Delta,\eta)$
of
% the
$(\Delta_0\vdash s_0') \unif (\nabla_0\vdash
% \pi\cdot
l_0)$, say $(\Delta_1,\theta_0)$. Thus (H3) $\Delta_1\vdash \Delta_0\theta_0, \nabla_0\theta_0, s_0'\theta_0 \ealeq
% \pi\cdot
l_0\theta_0$.
Define $s_1$ such that $\Delta_1\vdash (\context{C}_0'[
% \pi\cdot
r_0])\theta_0 \aleq s_1$. This together with (H3) give us the nominal \theory{R,E}-narrowing step $(\Delta_0\vdash s_0) \enarrow{R,E}^{\theta_0} (\Delta_1\vdash s_1)$.

Since $(\Delta_1,\theta_0) \leq_\theory{E} (\Delta,\eta)$, then there exists a substitution $\rho'$ such that $\forall X\in \mathcal{X}, \Delta\vdash X\theta_0\rho' \ealeq X\eta$ and (H4)~$\Delta\vdash \Delta_1\rho'$. Hence, $\Delta\vdash \theta_0\rho' \ealeq \eta$.
Because $\eta = \rho_0\cup\sigma$ and $\dom{\sigma}\cap V_0 = \emptyset$, $\rho_0$ is actually such that
% {\bf (H5)}
$\Delta\vdash \rho_0 \ealeq \theta_0\rho'|_{V_0}$. However, from hypothesis, $\rho_0 = \theta_0\rho_1$, which gives  (H5)~$\Delta\vdash \rho_1 \ealeq \rho'|_{V_0}$. Thus, item~\ref{cond3} $\Delta\vdash \rho_0|_V \ealeq \theta_0\rho_1|_V$ holds.

Since $\dom{\rho_0}\subseteq V_0$ and $\rho_0 = \theta_0\rho_1$, taking $V_1 = (V_0 \cup \vran{\theta_0}) - \dom{\theta_0}$ we have~\ref{cond4}: $\dom{\rho_1}\subseteq V_1$.
Therefore, from (H4) and  (H5), we get~\ref{cond1}: $\Delta\vdash \Delta_1\rho_1$.
Finally, one the one hand  we have  $\Delta \vdash t_1 \aleq \context{C}_0[
% \pi\cdot
r_0\sigma] \aleq \context{C}_0[
% \pi\cdot
r_0\eta] \ealeq \context{C}_0[
% \pi\cdot
r_0\theta_0\rho'] \ealeq \context{C}_0[
% \pi\cdot
r_0\theta_0\rho_1]$. On the other hand $\Delta\vdash s_1\rho_1 \ealeq (\context{C}_0'[
% \pi\cdot
r_0])\theta_0\rho_1 \aleq \context{C}_0'\theta_0\rho_1[
% \pi\cdot
r_0\theta_0\rho_1] \ealeq \context{C}_0'\rho_0[
% \pi\cdot
r_0\theta_0\rho_1] \ealeq \context{C}_0[
% \pi\cdot
r_0\theta_0\rho_1]$ (the last equivalence comes from $s_0\rho_0 = t_0$). Therefore, we have condition~\ref{cond2}: $\Delta\vdash s_1\rho_1 \ealeq t_1$.

Observe that for the base case we do not need the hypothesis $\theory{R}$ is
% $\theory{E}$-convergent
\theory{E}-terminating, \theory{R,E} is \theory{E}-confluent
and $\to_{\theory{R,E}}$ is $\theory{E}$-coherent. In matter of fact, these conditions are necessary for the induction hypothesis.
    \item \textbf{Inductive Step:} Let $n>1$ and assume that the result holds for sequences of  $n-1$ rewriting steps.
    Then,
    $$\Delta \vdash t_0 = s_0\rho_0\rightarrow_{\theory{R,E}} t_1\rightarrow_{\theory{R,E}}\ldots \rightarrow_{\theory{R,E}} t_n=t_0{\downarrow}$$
    By applying the base case
    on the rewrite step $\Delta \vdash t_0 \to_{\theory{R,E}} t_1$,  we get that
    $(\Delta_{0} \vdash s_{0}) \rightsquigarrow_{\theory{R,E}}^{\theta_{0}} (\Delta_1 \vdash s_1),$
   where $\rho_0$ is a $\theory{R,E}$-normalised substitution that satisfies $\Delta_0$ with $\Delta$,
   % and
   $\Delta \vdash s_1\rho_1 = t_1' \ealeq t_1$, and $\Delta \vdash \rho_0|_V \ealeq \theta_0\rho_1|_V$.
   Consider the sequence to any of the normal forms of~$t_1$:
      $$\Delta \vdash  t_1\overbrace{\rightarrow_{\theory{R,E}}\ldots \rightarrow_{\theory{R,E}}}^{n-1} t_n=t_1{\downarrow_\theory{R,E}}
      $$
    By the induction hypothesis, there exists  a narrowing sequence
     %----
    $$ (\Delta_1 \vdash s_1) \rightsquigarrow^{\theta_1}_{\theory{R,E}}  \ldots \rightsquigarrow^{\theta_{n-1}}_{\theory{R,E}} (\Delta_n \vdash s_n) $$
    %%--
    with $\theta=\theta_1\ldots \theta_{n-1}$, a normalised substitution $\rho_n$ such that
    $\Delta\vdash \Delta_i\rho_i$,  (**)~$\Delta \vdash s_i\rho_i =t_i'\ealeq t_i$, for every $i$ and $\Delta \vdash \rho_0|_V \ealeq \theta\rho_n|_V$.

   Note that from (**), it follows that $\Delta\vdash s_n\rho_n=t_n'\ealeq t_n$. Since \theory{R} is \theory{E}-convergent and $\to_{\theory{R,E}}$ is \theory{E}-coherent, it follows from Theorem~\ref{prop:nom-e-coherence}, that all the normal forms  of $t_0$ are $\ealeq$-equivalent. That is,
   $\Delta\vdash t_n'\ealeq t_1\downarrow_\theory{R,E}\ealeq t_0\downarrow_\theory{R,E}=t_n.$
    Therefore, there exists  a nominal \theory{R,E}-narrowing sequence
  $$(\Delta_0\vdash s_0) \rightsquigarrow^{\theta_0}_{\theory{R,E}} (\Delta_1 \vdash s_1) \rightsquigarrow^{\theta_1}_{\theory{R,E}}  \ldots \rightsquigarrow^{\theta_{n-1}}_{\theory{R,E}} (\Delta_n \vdash s_n).$$
\end{itemize}
\end{proof}

As a consequence of Lemmas~\ref{theo:narrowtorewrite} and \ref{theo:rewritenarr} we obtain the following result:
% Again, in addition to General assumption
\begin{theorem}[$\theory{E}$-Lifting Theorem]\label{teo:Elifting}
% \dani{Assume \theory{E} is compatible with $\vdash$ and substitutions
% and that there exists a complete \theory{E}-unification algorithm.}
% Let $\theory{R}{\cup}\theory{E}$ be an ENRS such that $\theory{R}$ is  \theory{E}-terminating, \theory{R,E} is \theory{E}-confluent and $\to_{\theory{R,E}}$ is $\theory{E}$-coherent.
Let $\theory{R{\cup}E_\alpha}$ be well-structured  and $\to_{\theory{R,E}}$ be $\theory{E}$-coherent.
To each finite sequence of nominal \theory{R,E}-rewriting steps corresponds a finite sequence of nominal $\theory{E}$-narrowing steps, and vice versa.
\end{theorem}
Since there exists an algorithm for nominal \theory{C}-unification
%and \theory{C} is compatible with substitutions,
Theorem~\ref{teo:Elifting} can be applied whenever $\theory{R{\cup}C_\alpha}$  is $\theory{C}$-convergent and $\to_{\theory{R,C}}$ is  $\theory{C}$-coherent.
\begin{corollary}
    The \theory{C}-Nominal Lifting theorem holds.
\end{corollary}

\section{\theory{R,E}-narrowing for \texorpdfstring{\theory{R{\cup}E_\alpha}}{RuE}-Unification}\label{sec:Narr-for-Unif}

Considering a \emph{closed} and well-structured $\theory{T=R{\cup} E_\alpha}$,
we can solve a nominal \theory{T}-unification problem applying \emph{closed (nominal) narrowing}.

\begin{definition}[Closed \theory{R,E}-narrowing]\label{def:clsd-narr-modE}
The \emph{one-step closed \theory{R,E}-narrowing relation} $(\Delta\vdash s) \rightsquigarrow_{\theory{R,E}}^c (\Delta' \vdash t)$ is the least relation generated by the rule below, for $R = (\nabla\vdash l\rightarrow r)\in \theory{R}$, term-in-context $\Delta \vdash s$, $\nw{R}$ a freshened variant of $R$ (so fresh for $R, \Delta, s, t$), position $\context{C}$, term $s'$,
%permutation $\pi$,
and substitution $\theta$
    \begin{prooftree}
    \AxiomC{$s\equiv \context{C}[s']$}
    \AxiomC{$\Delta', A(\nw{R})\# V(\Delta, s,t)\vdash \big( \nw{\nabla}\theta, \Delta\theta ,\ s'\theta \ealeq \nw{l}\theta, \ (\context{C}[\nw{r}])\theta \aleq t\big)$}
    \BinaryInfC{$(\Delta\vdash s) \rightsquigarrow_{\theory{R,E}}^c (\Delta' \vdash  t)$}
    \end{prooftree}

The \emph{closed nominal \theory{R,E}-narrowing relation}
%$(\Delta \vdash s) \rightsquigarrow_{\theory{R,E}}^{c} \cdots \rightsquigarrow_{\theory{R,E}}^{c} (\Delta' \vdash t)$
is the reflexive-transitive closure of the one-step closed \theory{R,E}-narrowing.
We denote it with  the composition $\rightsquigarrow_{\theory{R,E}}^{c}\ldots \rightsquigarrow_{\theory{R,E}}^{c}$ as usual.
\end{definition}

A ``closed \theory{E}-lifting'' theorem can be stated by replacing nominal \theory{R,E}-rewriting (\theory{R,E}-narrowing)
for {\em closed} nominal \theory{R,E}-rewriting (\theory{R,E}-narrowing).

\begin{theorem}[Closed \theory{E}-Lifting Theorem]\label{teo:closedElifting}
Let $\theory{R{\cup}E_\alpha}$ be closed and well-structured,
and $\to_{\theory{R,E}}$ be $\theory{E}$-coherent. To each finite sequence of closed nominal \theory{R,E}-rewriting steps corresponds a finite sequence of closed nominal $\theory{E}$-narrowing steps, and vice versa.
\end{theorem}

\subsection{Correctness of \texorpdfstring{\theory{R{\cup}E_\alpha}}{RuE}-unification  via narrowing}\label{ssec:unif_narr}

In order to find a solution for the unification problem $(\Delta \vdash s) \; _?{\overset{\theory{T}}{\approx}}_? \; (\nabla \vdash t)$, we will apply \emph{closed \theory{R,E}-narrowing} on $\Delta \vdash s$ and $\nabla \vdash t$ in parallel, that is, we \theory{R,E}-narrow a single term $u = (s,t)$ under $\Delta, \nabla$.

\begin{lemma}[Soundness]\label{lem:soundness}
Let
% $\theory{T=R{\cup}E}$
$\theory{T}=\theory{R}{\cup}\theory{E}_\alpha$ be
% a closed
% ENRS with an \theory{E}-convergent closed set of rules \theory{R}.
closed and well-structured.
Let $\Delta \vdash s$ and $\nabla \vdash t$ be two nominal terms-in-context and $$\Delta,\nabla \vdash (s,t) = u_0 \closenarrow \cdots \closenarrow \Delta_n \vdash u_n = (s_n, t_n)$$ a closed \theory{R,E}-narrowing derivation such that $(\Delta_n, \{s_n\ealeq t_n\})$ has an \theory{E}-solution, say $(\Gamma, \sigma)$.
Then $(\Gamma, \theta\sigma)$ is a $\theory{T}$-solution of the problem $(\Delta \vdash s) \; _?{\overset{\theory{T}}{\approx}}_? \; (\nabla \vdash t)$, where $\theta$ is the composition of substitutions along the \theory{R,E}-narrowing derivation.
\end{lemma}

\begin{proof}
Using Lemma~\ref{theo:narrowtorewrite}, with $\rho = {\tt Id}$, we can associate this \theory{R,E}-narrowing derivation with the following \theory{R,E}-rewriting derivation:
    $$\Gamma \vdash u_0\rho_0 \closerewrite u_1\rho_1 \closerewrite \cdots \closerewrite u_n\rho = u_n = (s_n,t_n)$$
Thus,
    ${\bf(I)}~ \Gamma \vdash s\rho_0 \starcloserewrite \ s_n\rho = s_n$ and
    $ {\bf(II)}~ \Gamma \vdash t\rho_0 \ {\starcloserewrite} \ t_n\rho = t_n$.
Because $(\Gamma,\sigma)$ is an \theory{E}-solution
of $(\Delta_n, {\tt Id}, \{s_n\ealeq t_n\})$, we have first that $\Gamma \vdash \Delta_n\sigma$, which from Lemma~\ref{lem:compatibility} gives us $\Gamma \vdash \Delta_0\theta\sigma$, that is, $\Gamma \vdash \Delta\theta\sigma$ and $\Gamma \vdash \nabla\theta\sigma$,
% \daniele{(Why $\theta$?)} \dani{because from the narrowing derivation we have $\Delta_1\vdash \Delta_0\theta_0, \Delta_2\vdash \Delta_1\theta_1 \vdash \Delta_0\theta_0\theta_1, \cdots \Delta_n\vdash \Delta_0\theta_0\cdots\theta_{n-1} = \Delta_0\theta$},
and second that ${\bf(III)} \ \Gamma\vdash s_n\sigma \ealeq t_n\sigma$.

Since we have $\theory{T=R{\cup}E}$,
% is a presented by \theory{R} (i.e., \theory{R} is a presentation of \theory{T}),
from {\bf (I)} and {\bf (II)}, and remembering that $\rho_0=\theta_0\theta_1\cdots\theta_{n-1}\rho = \theta\rho= \theta$, we get $\Gamma \vdash s\theta \taleq \cdots \taleq s_n$ and $\Gamma \vdash t\theta \taleq \cdots \taleq t_n$.
% \alert{$\Gamma \vdash s\theta = s_n \in \theory{T}$ and $\Gamma \vdash t\theta = t_n \in \theory{T}$?}.
From {\bf (III)}, it follows that $\Gamma \vdash (s\theta)\sigma \taleq s_n\sigma \ealeq t_n\sigma \taleq (t\theta)\sigma$.
Therefore, we have $\Gamma\vdash (\Delta\theta\sigma, \nabla\theta\sigma, s\theta\sigma \taleq t\theta\sigma)$ and hence, $(\Gamma, \theta\sigma)$ is a $\theory{T=R{\cup}E}$-solution to $(\Delta \vdash s) \; _?{\overset{\theory{T}}{\approx}}_? \; (\nabla \vdash t)$.
\end{proof}

The next result states that if we have a \theory{T}-solution between two terms-in-context, then we can build a closed \theory{R,E}-narrowing derivation from this starting terms, and thus there exists an \theory{E}-solution between the final terms.

\begin{lemma}[Completeness]\label{lem:completeness}
Let
% $\theory{T=R{\cup}E}$
$\theory{T=R}{\cup}\theory{E}_\alpha$ be
% a closed ENRS, presented by an \theory{E}-convergent set \theory{R} of closed rules,
closed and well-structured,
and let $\to_\theory{R,E}$ be \theory{E}-coherent.
Let $\Delta \vdash s$ and $\nabla \vdash t$ be two nominal terms-in-context, such that the problem $(\Delta \vdash s) \; _?{\overset{\theory{T}}{\approx}}_? \; (\nabla \vdash t)$ has a $\theory{T}$-solution, $(\Delta',\sigma)$, and let $V$ be a finite set of variables containing $V(\Delta, \nabla, s, t)$.
Then there exists a closed \theory{R,E}-narrowing derivation: $$\Delta,\nabla \vdash u = (s,t) \closenarrow \cdots \closenarrow \Gamma_n \vdash (s_n,t_n),$$ such that there exists an \theory{E}-solution $(\Gamma, \mu) $ for $(\Gamma_n,  \{s_n\ealeq t_n\})$  and such that $(\Gamma, \theta\mu) \leq^V_\theory{E} (\Delta', \sigma)$, where $\theta$ is the composition of the narrowing substitutions.
Moreover, we are allowed to restrict our attention to $\closenarrow$-derivations such that: $\forall i, 0\leq i < n$, $\theta_i|_V$ is normalised.
\end{lemma}
\begin{proof}
    By Definition~\ref{def:nominalEunif}, we have $\Delta' \vdash \Delta\sigma, \nabla\sigma, s\sigma \taleq t\sigma$.

    Take $\rho_0 = \sigma{\downarrow}$, that is, the normal form of $\sigma$ in $\Delta'$: $\Delta' \vdash X\rho_0 \ealeq (X\sigma){\downarrow}$, for all $X\in V$.
    Since the rules are closed, it follows that $\Delta' \vdash \Delta\rho_0, \nabla\rho_0, s\rho_0 \taleq t\rho_0$.

    Since
    % $\theory{T=R{\cup}E}$
    $\theory{T=R}{\cup}\theory{E}_\alpha$ is
    % a closed ENRS, presented by an \theory{E}-convergent rewrite system \theory{R},
    closed and well-structured,
    $\Delta' \vdash s\rho_0 \taleq t\rho_0$ gives us that $s\rho_0$ and $t\rho_0$ have the same normal form in $\Delta'$, which we will call $r$. Then
    $$\Delta' \vdash u\rho_0 = (s\rho_0,t\rho_0) \closerewrite \cdots \closerewrite (r,r) = (u\rho_0){\downarrow}$$
    By the Lemma~\ref{theo:rewritenarr}, there exists a corresponding $\closenarrow$-derivation ending with $\Gamma_n \vdash (s_n,t_n)$ such that
    $\Delta'\vdash (s_n\rho_n, t_n\rho_n) = (s_n\rho, t_n\rho) \ealeq (u\rho_0){\downarrow} = (r,r)$
    and $\Delta'\vdash \Gamma_n\rho_n = \Gamma_n\rho$.
    Thus, $(\Delta', \rho)$ is an \theory{E}-solution of $(\Gamma_n,\{s_n\ealeq t_n\})$.

    Since $(\Gamma,\mu)$ is the least \theory{E}-unifier, it follows that $(\Gamma,\mu) \leq_\theory{E} (\Delta',\rho)$ and there exists a substitution $\delta$ such that $\Delta' \vdash X\mu\delta \ealeq X\rho$, for all $X\in V$ and $\Delta'\vdash \Gamma\delta$.
    Therefore, by Lemma~\ref{theo:rewritenarr}, $\Delta' \vdash \rho_0|_V \aleq \theta\rho|_V \aleq (\theta\mu\delta)|_V$ and $\Delta' \vdash \rho_0|_V \aleq \sigma|_V$, which gives us $(\Gamma,\theta\mu) \leq_\theory{E}^V (\Delta', \sigma)$.
\end{proof}

Now we can describe how to build a complete set of
% $\theory{T=R{\cup}E}$
$\theory{T=R}{\cup}\theory{E}_\alpha$-unifiers for two terms-in-context, which is an extension of the result initially proposed in~\citep{NominalNarrowing16}.

\begin{theorem}\label{teo:complet_clsd}
Let
% $\theory{T=R{\cup}E}$
$\theory{T=R}{\cup}\theory{E}_\alpha$ be
% a \dani{closed} ENRS, presented by an \theory{E}-
% % convergent
% \dani{terminating} nominal rewrite theory \theory{R},
% \dani{let \theory{R,E} be \theory{E}-confluent},
closed and well-structured,
and let $\to_{\theory{R,E}}$ be $\theory{E}$-coherent. Let $\Delta \vdash s$ and $\nabla \vdash t$ be two terms-in-context, and $V$ be a finite set of variables containing $V(\Delta, s, \nabla, t)$. Let $\mathcal{S}$ be the set of pairs $(\Psi, \sigma)$ such that there exists a $\closenarrow$-derivation:
    $$\Gamma_0 \equiv \Delta, \nabla \vdash u = (s, t) = u_0 \closenarrow \cdots \closenarrow \Gamma_n \equiv \Delta_n, \nabla_n \vdash u_n = (s_n, t_n),$$
where $(\Gamma_n, \{s_n\ealeq t_n\})$ has a complete set of \theory{E}-solutions $S' = \mathcal{U}_\theory{E}(\Delta_n \vdash s_n, \nabla_n \vdash t_n)$, $(\Psi,\mu)\in S'$,
$ \sigma \equiv \theta\mu$, and
$\theta$ is the normalised composition of the \theory{R,E}-narrowing substitutions. Then $\mathcal{S}$ is a complete set of \theory{T}-unifiers of $\Delta \vdash s$ and $\nabla \vdash t$ away
from $V$.
\end{theorem}

%\MF{Shouldn't we say that the theorem is a consequence of the Soundness and Completeness lemmas?? I see the proof has been commented out (??)}
 \begin{proof}
     Consequence of Lemmas~\ref{lem:soundness} and~\ref{lem:completeness}.
 \end{proof}

\paragraph{\bf Unification procedure}
Theorem~\ref{teo:complet_clsd} says that if we have a narrowing derivation ending in a pair of terms-in-context that has an \theory{E}-unifier, say $(\Psi,\mu)$, then if we compose $\mu$
with $\theta$, the composition of narrowing substitutions, we get $\sigma$ and hence we have a \theory{T}-solution $(\Psi,\sigma)$ for the pair of terms-in-context at the beginning of the %narrowing
derivation.

\begin{remark}
If \theory{E} is infinitary or nullary, each node in the narrowing tree may have infinite children. This is the case, for instance, of the commutativity theory \theory{C} that may yield an infinite set of solutions due to fixed-point problems. So we have a tree that is possibly infinite both in width and depth, and therefore, a standard breadth-first-search will not yield a complete unification procedure.
 However, notice that even if a node may have an infinite number of children, these are \emph{enumerable}. It is well-known that an enumerable union of enumerable sets is enumerable~\citep{Ebbinghaus:mathematicalLogic}; therefore, we can enumerate all nodes in the narrowing tree, yielding a sound and complete \theory{T}-unification procedure also in this case.

 If \theory{E} is finitary, the narrowing tree for $(\Delta \vdash s) \; _?{\overset{\theory{T}}{\approx}}_? \; (\nabla \vdash t)$ is finitely branching (even if it might be infinite in depth), and we can enumerate all the nodes using a standard breadth-first-search strategy (i.e., building the tree by levels).
 \end{remark}

\subsection{Example: Symbolic Differentiation with Commutativity (\theory{E=C})}\label{ssec:exam_diff}

    % The rewrite rules in Example~\ref{ex:lambda} define a $\lambda$-calculus with names and explicit substitutions~\cite{FernandezGM04_NRSystems}, the extension with numbers and operations ($plus$, $mult$, $sin$, $cos$) is straightforward.
   % The extension with numbers and operations ($\plus$, $\mult$, $\sin$, $\cos$) of rules in Example~\ref{ex:lambda}. The rules are presented in Figure~\ref{fig:diff_rules}
    % define a $\lambda$-calculus with names and explicit substitutions~\cite{FernandezGM04_NRSystems}, the extension with numbers and operations ($plus$, $mult$, $sin$, $cos$)
    %is straightforward.

    In this section we consider symbolic differentiation~\citep{prehofer1998solving} to illustrate our results. We consider the term $\f{diff}(F, X)$ that specifies the computation of the differential of a function $F$ at a point $X$. Here, $F$ is a meta-level unknown that can be instantiated by a term.  Differentiation is implemented by  the
    rewrite rules $\theory{R}_{\f{diff}}$ presented in Figure~\ref{fig:diff_rules}.
We consider the theory  \theory{T=R_{\f{sub}} \cup R_\f{diff}\cup C_\alpha}
defined by the union of the rules in \theory{R_\f{sub}}, defined in Example~\ref{ex:lambda},
together with the rules
$\theory{R}_{\f{diff}}$,
and \theory{C} being the commutative axioms for the binary function symbols  $\{\plus, \mult\}$.

\begin{figure}[!t]
\noindent

\resizebox{\textwidth}{!}{
$
    \begin{array}{lclcl}
    & \vdash & \plus(0,X) & \to & X  \\
    & \vdash & \plus(X,s(Y)) & \to & s(\plus(X,Y))  \\
    & \vdash & \mult(X,1) & \to & X\\
    y\# F & \vdash & \f{diff}(\texttt{lam}([y]F),X) & \to & 0  \\
    & \vdash & \f{diff}(\texttt{lam}([y]y),X) & \to & 1  \\
    & \vdash & \f{diff}(\texttt{lam}([y]\sin(F)),X) & \to & \mult(\cos(\texttt{sub}([y]F,X)),\f{diff}(\texttt{lam}([y]F),X) )  \\
    & \vdash & \f{diff}(\texttt{lam}([y]\plus(F,G)),X) & \to & \plus(\f{diff}(\texttt{lam}([y]F),X) ,\f{diff}(\texttt{lam}([y]G),X) )  \\
    & \vdash & \f{diff}(\texttt{lam}([y]\mult(F,G)),X) & \to & \plus(\mult(\f{diff}(\texttt{lam}([y]F),X),\texttt{sub}([y]G,X) ),\\
    & & & &\hfill \mult(\f{diff}(\texttt{lam}([y]G),X),\texttt{sub}([y]F,X) ) ) \\
\end{array}
$}
\caption{$\theory{R}_\f{diff}$ -- Rules for Symbolic Differentiation Modulo Commutativity}\label{fig:diff_rules}
\end{figure}

This system is closed
% \daniele{Why?}\alert{not \theory{C}-convergent}, so \theory{C}-narrowing is not necessarily complete; however, we can still
and \theory{C}-convergent, so Lemma~\ref{lem:soundness} applies and we can use \theory{C}-narrowing to
obtain a solution $(\{y\#G\}, \sigma=[F\mapsto y])$ for the \theory{T}-unification problem
\vspace{-2mm}
$$\big(y\#G\vdash \texttt{lam}([z]\f{diff}(\texttt{lam}([y]\plus(\sin(F),G)),z)) \big) \overset{\theory{T}}{\unif} \big( \vdash \texttt{lam}([z]\cos(z))\big).$$

%This problem verifies that
%trying to see
% if computing the differential of $plus(sin(F),G)$ with relation to $y$ at the point $z$ gives the same result as computing the differential of the $cos(z)$, when we know that $y$ cannot occur free in $G$.
Figure~\ref{fig:ex_diff} illustrates the use of \theory{C}-narrowing to solve this problem. For simplicity, we omit the freshness contexts.

\begin{figure}[!t]
\begin{center}
\resizebox{\textwidth}{!}{
\begin{tikzpicture}[node distance=1.5cm, every node/.style={anchor=base west}]
    \node (A) {$\texttt{lam}([z]\teal{\f{diff}(\texttt{lam}([y]\plus(\sin(F),G)),z)}) \overset{\theory{T}}{\unif} \texttt{lam}([z]\cos(z))$};
    \node (B) [below of=A] {$\texttt{lam}([z]\plus(\teal{\f{diff}(\texttt{lam}([y'](y\ y')\cdot G)},z), \f{diff}(\texttt{lam}([y'](y\ y')\cdot \sin(F)),z) ) ) \overset{\theory{T}}{\unif} \texttt{lam}([z]\cos(z))$};
    \node (C) [below of=B] {$\texttt{lam}([z]\teal{\plus(0,\f{diff}(\texttt{lam}([y'](y\ y')\cdot \sin(F)),z) )} ) \overset{\theory{T}}{\unif} \texttt{lam}([z]\cos(z))$};
    \node (D) [below of=C] {$\texttt{lam}([z]\teal{\f{diff}(\texttt{lam}([y'](y\ y')\cdot \sin(F)),z)} ) \overset{\theory{T}}{\unif} \texttt{lam}([z]\cos(z))$};
    \node (E) [below of=D] {$\texttt{lam}([z]\mult(\cos(\texttt{sub}([y'']\pi\act F,z)),\teal{\f{diff}(\texttt{lam}([y'']\pi\act F),z) })) \overset{\theory{T}}{\unif} \texttt{lam}([z]\cos(z))$};
    \node (F) [below of=E] {$\texttt{lam}([z]\mult(\cos(\teal{\texttt{sub}([y'']y'',z)}),1)) \overset{\theory{T}}{\unif} \texttt{lam}([z]\cos(z))$};
    \node (G) [below of=F] {$\texttt{lam}([z]\teal{\mult(\cos(z),1)}) \overset{\theory{T}}{\unif} \texttt{lam}([z]\cos(z))$};
    \node (H) [below of=G] {$\texttt{lam}([z]\cos(z)) \overset{\theory{T}}{\unif} \texttt{lam}([z]\cos(z))$};
    \draw[->,decorate,decoration={snake}] (A) -- node[right] {$\theta_1$} (B);
    \draw[->,decorate,decoration={snake}] (B) -- node[right] {$\theta_2$} (C);
    \draw[->,decorate,decoration={snake}] (C) -- node[right] {$\theta_3$} (D);
    \draw[->,decorate,decoration={snake}] (D) -- node[right] {$\theta_4$} (E);
    \draw[->,decorate,decoration={snake}] (E) -- node[right] {$\theta_5$} (F);
    \draw[->,decorate,decoration={snake}] (F) -- node[right] {$\theta_6$} (G);
    \draw[->,decorate,decoration={snake}] (G) -- node[right] {$\theta_7$} (H);
\end{tikzpicture}
}
\end{center}
\caption{\theory{T}-unification narrowing tree for
% Example~\ref{ex:diff}
the nominal \theory{C}-unification problem using $\theory{R}_\f{diff}$}
    \label{fig:ex_diff}
\end{figure}

The first closed \theory{C}-narrowing step uses a freshened version of the fourth rule of differentiation\footnote{$\vdash \f{diff}(\texttt{lam}([y']\plus(F',G')),X') \to \plus(\f{diff}(\texttt{lam}([y']F'),X'),\f{diff}(\texttt{lam}([y']G'),X') )  $} with the assumptions $y'\#F,G$. The substitution used is $\theta_1= [F'\mapsto (y\ y')\cdot \sin(F), G'\mapsto (y\ y')\cdot G, X'\mapsto z]$.

The commutativity of $\plus$ is used before instantiating the subterm with the narrowing substitution.
Since the initial problem contains the constraint $y\# G$, we can  narrow, using rule $  y\# F \vdash \f{diff}(\texttt{lam}([y]F),X) \to 0$ (second narrowing step)
and from rule $\plus(0,X) \to X$, we can rewrite/narrow again (third narrowing step).
The next narrowing step uses the  freshened rule:
{\small $$\vdash \f{diff}(\texttt{lam} ([y''] \sin(F'')), X'') \rightarrow
\mult(\cos(\texttt{sub}([y'']F'',X'')),\f{diff}(\texttt{lam}([y''] F''),X''))$$ }

\noindent with the assumption $y''\#F$.
The substitution is $\theta_4=[F'' \mapsto \pi\act F, X'' \mapsto z]$), with  the permutation $\pi=(y'\ y'')(y\ y')$.

Then we use the freshened rule $\vdash \f{diff}(\texttt{lam}([w]w), W) \rightarrow 1$ with substitution $\theta_5=\{F \mapsto y, W \mapsto z\}$ and assumption $w \#  F$ to narrow the second argument of $\mult$.
Now we can use the rules for $\texttt{sub}$, and rewrite (hence also narrow) to the next level, and rewriting (narrowing) with $\vdash \mult(X,1)\to X$, we obtain two equal terms.

%\end{example}

\section{Refinement: Nominal Basic \textit{Closed} Narrowing}
\label{sec:basnarr}

In the previous section, we saw that the closed narrowing relation $\closenarrow$ provides a sound and complete $(\theory{R\cup\alpha\cup E})$-unification procedure. As in the first-order approach, narrowing may be non-terminating, meaning that this nominal $(\theory{R\cup\alpha\cup E})$-unification procedure might not  terminate (see Example~\ref{exa:narrow-infty}), even when only a finite number of unifiers exist~\citep{VariantNarrowing:EscobarMS09}.

In this section, we introduce a refinement of nominal closed narrowing - called {\em (nominal) basic closed narrowing} - designed to manage the potentially infinite number of narrowing derivations when $\theory{E} = \emptyset$.
% Basic closed narrowing is used to eliminate redundant closed narrowing derivations in order to give sufficient conditions for the termination of the narrowing process.
In a basic derivation, narrowing is never applied to a subterm introduced by a previous narrowing substitution~\cite{NominalNarrowing16, ALPUENTE_TerminationofNarrowing,BasicNarrowing:MiddeldorpH94}. The difference here is that we work with closed rewriting/narrowing derivations taking into account theories $\theory{T}=\theory{R}\cup\alpha$, that can be oriented into a convergent NRS $\theory{R}$ with embedded $\alpha$-equivalence.

We conclude the section with a discussion of the narrowing technique for the case $\theory{E}\neq \emptyset $, we highlight the limitations of the current approach and outline directions for future work.

\subsection{Basic Closed Narrowing for \texorpdfstring{$\theory{E}=\emptyset$}{E=emptyset} }

Let $\theory{T}=\theory{R}\cup\alpha$, be an equational theory
% that has a complete presentation as
with
a convergent closed nominal rewrite system $\theory{R}$.
In the rest of this section, we assume $\theory R = \{R_k\equiv \nabla \cent l_k
\rightarrow r_k \}$, for $k\geq 1$, and we write ${\cal P}os(t)$ to denote the set of positions of the term $t$ and $\overline{{\cal P}os}(t)$ to denote the set of \emph{ground} positions of the term $t$.

\begin{definition}[Basic Rewriting] \label{defi:basic.positions}
 Let $ s$ be a  term and
 $U=\overline{\mathcal{P}os}
 (r)$, for some subterm $ r$ of $ s$.
 A  rewriting
 derivation
 $$\Delta \cent s=s_0\rightarrow_{[\context{C}_0,R_0]}s_1 \rightarrow_{[\context{C}_1,R_1]}
 \rightarrow \ldots \rightarrow_{[\context{C}_{n-1},R_{n-1}]} s_n $$
is \emph{based on $U$} and construct
 sets of positions $U_i\subset \mathcal{P}os(s_i)$, $0\leq i \leq n$,
  inductively, as follows:
  \begin{enumerate}
  \item the empty derivation is
 based on $U$, and $U_0=U$;
 \item if a derivation up to $s_i$ is based on $U$,
 then the derivation obtained from it by adding one step
 $s_i\rightarrow_{[\context{C}_i,R_{i}]}s_{i+1}$ is based on $U$ iff $\context{C}_i\in
 U_i$, and in this case we take:
 $$U_{i+1}\!=\!(U_i-\{\context{C}\in U_i \mid
 \context{C}_i\leq \context{C}\})\cup \{\context{C}_i.\context{C} \mid  \context{C}\in \overline{\mathcal{P}os}(r_{i})\},$$
 where $r_i$ denotes the right-hand side of the rule $R_i$  in
 $\theory R$\footnote{Given $\context{C}_i=(s_i, \_)$ and $\context{C}=(s, \_)$, it follows $\context{C}_i.\context{C}=(s_i\{\_ \mapsto s\}, \_)$  and $\context{C}_i\leq \context{C}$ if $\exists t : \, s_i\{\_ \mapsto t\}=s$. }.
 \end{enumerate}
Positions in $U_i$ are called \emph{basic} and positions in $\overline{{\cal P}os}(r) - U_i$ are  \emph{non-basic}.
%We say the derivation is \emph{basic} if $\context{C}_i \in U_i$ for $1\leq i < n$.
\end{definition}
\normalsize

A (closed) rewrite step $\Delta \vdash \context{C}[s] \to \context{C}[s']$ at position $\context{C}$ is \emph{innermost} if for any $\context{C}_i$ such that $\context{C} < \context{C}_i$ and $\context{C}[s] = \context{C}_i[s_i]$,
there is no (closed) rewrite step $\Delta \vdash \context{C}_i[s_i] \to \context{C}_i[t]$ at position $\context{C}_i$. In other words, there is no (closed) rewrite step inside $s$.
An innermost (closed) rewrite derivation contains only innermost (closed) rewrite steps.

The next result is the nominal version of the well-known {\em Hullot's Property}~\cite{Hullot80} for closed rewriting, which is fundamental to obtain a complete and finite $\theory{T}$-unification algorithm (Theorem~\ref{Thm:complete_basic}) using basic closed narrowing.

\begin{lemma}[Hullot's Property]\label{lema:based}
Let $\Delta \cent s\approx_{\alpha}s_0\rho_0$, with $\rho_0$ normalised in $\Delta$ w.r.t. $\rew$. Every innermost (closed) nominal rewrite derivation issuing from $\Delta\cent s$ is based on $\overline{\mathcal{P}os}(s_0)$.
\end{lemma}

\begin{proof}
We will prove the lemma for the case of closed rewriting derivations. The proof of the case without the closedness assumption is analogous.
Consider the following innermost closed nominal rewrite derivation:
$$\Delta\vdash s_0\rho_0 \aleq s = t_0 \to^c_{[\context{C}_0,R_0]} \cdots \to^c_{[\context{C}_{n-1},R_{n-1}]} t_n$$
Let $U = \overline{{\cal P}os}(s_0)$ and consider positions $U_0,U_1 \ldots, U_{n-1}$ as in Definition~\ref{defi:basic.positions}.

Take one arbitrary step  $\Delta\vdash t_i\to^c_{[\context{C}_i,R_i]} t_{i+1}$ and suppose by contradiction that $\context{C}_i$ is not based on $U_i$, i.e., that $\context{C}_i \in {\cal P}os(t_i) - U_i$.
% , and reach a contradiction.
There are two cases to consider:
\begin{itemize}
    \item $\context{C}_i$ is a position that was introduced by $\rho_0$. Such reduction is not possible since, by hypothesis, $\rho_0$ is normalised.
    \item $\context{C}_i$ is a position that was introduced by an instance of the rhs of
    a rule, consider w.l.o.g
    the rule $R_{i-1}\equiv \nabla_{i-1}\vdash l_{i-1}\to r_{i-1}$ that was applied in the previous step $\Delta\vdash t_{i-1}\to^c_{[\context{C}_{i-1},R_{i-1}]} t_{i}$. That is, for some substitution $\theta$ and $t_{i-1}\equiv \context{C}_{i-1}[t'_{i-1}]$, the following holds:
        \begin{mathpar}
    \inferrule{ \Delta, A(\nw{R}_{i-1})\# V(\Delta, t_{i-1},t_{i})\vdash \big( \nw{\nabla}_{i-1}\theta,\ t_{i-1}' \aleq \nw{l}_{i-1}\theta, \ \context{C}_{i-1}[\nw{r}_{i-1}\theta] \aleq t_{i}\big)}
    {\Delta\vdash t_{i-1}\rightarrow_{\theory{R}}^{c}\  t_{i} }
    \end{mathpar}
    Note that $t_i$ has positions that were introduced by $\nw{r}_{i-1}\theta$ via $\theta$-instances of variables of $\nw{r}_{i-1}$. Those positions are not in $U_i$.  Suppose that $\context{C}_i$ is one of such positions of $t_i$, i.e., $t_i=\context{C}_i[t_i']$ and $t_i'$ is a redex for rule $R_i$. However, $V(\nw{r}_{i-1})\subseteq V(\nw{l}_{i-1})$ by definition. Therefore, this position also occurs in $\nw{l}_{i-1}\theta$, and consequently, in $t_{i-1}'$. This implies that there is a redex in a position below $\context{C}_{i-1}$, and this contradicts the fact that the derivation is innermost.
\end{itemize}
\end{proof}

In a basic closed narrowing derivation, a closed narrowing step is never applied to a subterm introduced by the substitution of a previous closed narrowing step. Formally,
\begin{definition}[Basic closed narrowing]
A closed  narrowing derivation
$$(\Delta_0\vdash s_0)\cnarr{\context{C}_0,R_{0},\theta_0}\ldots \rightsquigarrow^c_{[\context{C}_{i-1},R_{i-1},\theta_{i-1}]}(\Delta_i\vdash s_i),
$$
is  \emph{basic} if it is based on $\overline{\mathcal{P}os}{(s_0)}$ (in the same sense as in Definition~\ref{defi:basic.positions}).
\end{definition}

\begin{example}\label{exa:narrow-infty}
    Let $\theory{R} = \{a\#X\vdash \forall[a]\forall[a]X \to X\}$ be an ENRS. The infinite sequence
    \[
    \begin{array}{cll}
       \; \vdash (\forall[a]X,X)  &  \narrow^c_{[X\mapsto \forall[a]X']}& \;\vdash (X',\forall[a]X')\\
         & \narrow^c_{[X'\mapsto \forall[a]X'']}& \;\vdash (\forall[a]X'',X'')\\
         &\narrow^c_{[X''\mapsto \forall[a]X''']}& \cdots\\
    \end{array}
    \]
    % $$\;\vdash (\forall[a]X,X) \narrow^c_{[X\mapsto \forall[a]X']} (X',\forall[a]X') \narrow^c_{[X'\mapsto \forall[a]X'']} (\forall[a]X'',X'') \narrow^c_{[X''\mapsto \forall[a]X''']} \cdots$$
    is the only narrowing derivation issued from $(\forall[a]X,X)$. It is non-basic: $U_1= \{[\_]\}$ but $\context{C}_1= \forall[a][\_]$.
\end{example}

The next result guarantees that in the case of $\theory{E}=\emptyset$, the nominal (closed) lifting theorems leverage basic (closed) narrowing derivations.

\begin{lemma}\label{teo:basicNar}
Suppose that $\theory{E}=\emptyset$. Then the narrowing derivations constructed in the Theorem~\ref{teo:Elifting} and Theorem~\ref{teo:closedElifting}
are all basic.
\end{lemma}

\begin{proof}
Let \theory{E = \emptyset} be an empty theory and  \theory{R}  be terminating and confluent closed nominal rewrite system. The coherence property is trivially achieved for empty theories.
 Consider the nominal narrowing derivation of Theorem~\ref{teo:Elifting}:
$$(\Delta_0\vdash s_0)\rightsquigarrow_{[\context{C}'_0,R_0,\theta_0]}   \ldots  \rightsquigarrow_{[\context{C}'_{n-1},R_{n-1},\theta_{n-1}]}(\Delta_n\cent s_n)$$
 with the associated rewriting sequence $\Delta\cent t_0 = s_0\rho_0\rightarrow_{[\context{C}_0,R_0]}  \ldots \rightarrow_{[\context{C}_{n-1},R_{n-1}]} t_n = t_0{\downarrow}$,
such that $\rho_0$ is normalised satisfying $\Delta_0$ with $\Delta$. Since $\theory{R}$ is confluent we may assume that the nominal rewriting sequence from $\Delta \cent s_0\rho_0$ is innermost. By
Lemma~\ref{lema:based}, this nominal rewriting derivation is  based on $\overline{\mathcal{P}os}(s_0)$, and since the sets
$U_i$ in the two derivations are equivalent, it follows that the considered nominal narrowing derivation is basic. The verification for Theorem~\ref{teo:closedElifting} is analogous.
%MF: the line above said that the sets U_i are the same in both derivations but they are not exactly the same with our defs, I put the word equivalent,  I think it is clear that there is a direct correspondence.
\end{proof}

The main interest of basic closed narrowing is that we can give a sufficient condition for the termination of the narrowing process when we consider only basic $\rightsquigarrow^c$-derivations and therefore for the termination of the corresponding nominal \theory{T}-unification procedure.

\begin{theorem}\label{thm:termination}
Let $\theory{R}=\{\nabla_k\cent l_k \rightarrow r_k\}$ be a convergent closed nominal rewriting system such that any basic $\rightsquigarrow^c$-derivation
issuing from any of the right-hand sides $r_k$ terminates. Then any basic $\rightsquigarrow^c$-derivation issuing from any nominal term terminates.
\end{theorem}
\begin{proof}
    Consider the following basic $\narrow^c$-derivation:
    $$\Delta\vdash s=s_0 \narrow^c_{[\context{C}_0,R_0]} s_1 \narrow^c \cdots \narrow^c_{[\context{C}_{n-1},R_{n-1}]} s_n$$
     The first  step gives us that $\context{C}_0\in \overline{\mathcal{P}os}{(s_0)}$. If the second step issues from an rhs $r_k$, then by hypothesis, the derivation terminates. Suppose then, w.l.o.g, that the second step does not start from an rhs. Since our derivation is basic, we know that $\context{C}_1\in \overline{\mathcal{P}os}{(s_0)}-\{\context{C}_0\}$. We repeat the reasoning process. Since $\overline{\mathcal{P}os}{(s_0)}$ is finite, we know that eventually we would not be able to apply a step unless it is in a rhs $r_k$, which gives us the result.
\end{proof}

The next results show that we can build a complete set of $\theory{T=R{\cup}E}$-unifiers for two terms-in-context using only basic $\narrow$-derivations.
\begin{theorem}\label{teo:basic_complet_clsd}
If  \theory{E=\emptyset} then Theorem~\ref{teo:complet_clsd} holds for closed basic narrowing.
\end{theorem}
\begin{proof}
% Let \theory{E = \emptyset} be an empty theory and  \theory{R}  be terminating and confluent closed nominal rewrite system.
Let \theory{E = \emptyset} be an empty theory and  \theory{R}  be terminating and confluent closed nominal rewrite system. The coherence property is trivially achieved for empty theories.
Consider the closed narrowing derivation from Theorem~\ref{teo:complet_clsd}
$$
% \Gamma_0 \equiv
\Delta, \nabla
\vdash u = (s, t) = u_0
% \underset{[\context{C}_0,R_0]}{\closenarrow}
\narrow^c_{[\context{C}_0,R_0]}
\cdots
% \underset{[\context{C}_{n-1},R_{n-1}]}{\closenarrow}
\narrow^c_{[\context{C}_{n-1},R_{n-1}]}
% \Gamma_n \equiv
\Delta_n, \nabla_n \vdash u_n = (s_n, t_n),$$
The proof follows if all the $\context{C}_i$, $i=1\ldots n-1$, along the derivation are basic.

\end{proof}

\begin{theorem}\label{Thm:complete_basic}
    If $\theory{R}$ is closed and satisfies the hypothesis of Theorem~\ref{thm:termination}, the construction of Theorem~\ref{teo:basic_complet_clsd} leads to a complete and finite \theory{T}-unification algorithm.
\end{theorem}

\subsection{Nominal Basic Narrowing for \texorpdfstring{$\theory{E}\neq \emptyset$}{E=/emptyset}}\label{ssec:basic_equational}

The nominal language includes the first-order language of terms, as expected, all the problems that exist in reasoning modulo equational theories in the former, are still problems in the latter. In addition to those, one has to consider the extra complications that languages with binders include, such as dealing with freshness and modulo renaming.
%It is a well-known result that first-order basic narrowing modulo \theory{E} fails to terminate
%(Example~\ref{exa:non-termination})
% for many equational theories, and it may not generate a complete set of unifiers (Example~\ref{exa:incomplete-ac}).
The next example (taken from~\citep{CLD05_TheFiniteVariantProperty}) and expressed in the nominal framework using empty freshness constraints -- shows that Hullot's property does not necessarily hold when $\theory{E}\neq \emptyset$. This  compromises the completeness of the unification procedure.

\begin{example}\label{exa:incomplete-ac}
Let $R=\{X+0 \to X, X+X\to 0, X+(X+Y)\to Y\}$ be  a  nominal rewrite system, which is known to be $\theory{AC}$-convergent. Let us consider the term $s_0=[c](X_1+X_2)$ and the normalised substitution $\rho_0=[X_1\mapsto (a+b), X_2\mapsto (a+b)]$. The following innermost rewriting $\to_{\theory{R,AC}}$-derivation starting from $s_0\rho_0$:
% \dani{put an abstraction}
%
%$\theta=\{X/a, Y/b+b\}$
%$l\theta= a + (a + (b+b))=_{AC} (a+b)+(a+b)$.
$$[c]((a+b)+(a+b))\to [c](b + b) \to [c]0$$
is not based on $s_0$. Contradicting Hullot's property (Lemma~\ref{lema:based}). In fact, the first step takes place at position $
% \epsilon
[c][\_]
\in U_0=\overline{{\cal P}}os(s_0)=\{
% \epsilon
[\_],[c][\_]
\}$, by applying the rewrite rule $X+X+Y\to Y$ by rearranging the arguments (modulo \theory{AC}). The new set of positions is $U_1=
% \emptyset
\{[\_]\}$, since there are no ground positions in the rhs of the rule applied. The second derivation occurs in a position that is not from the set of basic positions in $s_0$.

\end{example}

The point here is that the notion of positions is static and we are working with representatives of \theory{AC}-equivalence classes of terms and positions that have different subterms depending on the representative, and this suggests a more {\em dynamic} notion of positions. Looking at the
% root
position $\context{C}\equiv[c][\_]$ of the representative $[c]((a+b)+(a+b))$, there are no (strictly) internal redexes. The application of rule $X+X+Y\to Y$ introduces the redex $[c](b+b)$. However, if we take another representative of $s_0\rho_0$, say $[c](a+( a+(b+b)))$, now the derivation starting at the
% root
position $\context{C}$, applying the same rule, is not innermost anymore.

Breaking the Hullot's Property affects the completeness of the algorithm for \theory{AC}-unification procedure using \theory{AC}-narrowing, as the original proof of completeness - in the first-order approach - relies on the property. The existence of a finite minimal complete set of \theory{AC}-unifiers  for unification problems modulo \theory{AC} is well-known~\cite{DBLP:journals/jacm/Stickel81,DBLP:journals/jar/AyalaRinconFSKN24}.

Basic closed narrowing  works modulo $\alpha$-equivalence classes because the elements in the class are not structurally affected by renaming, there are no structural changes in the terms and the subterms at each position except for the renamed atoms - and this does not affect the proof of Lemma~\ref{lema:based}. However,  dealing with theories including $\theory{AC}\cup \alpha$ will be as problematic in the nominal framework as it is in the first-order approaches.

\paragraph{Alternative Methods and Current Limitations}
There are methods addressing both the non-termination of basic narrowing and incompleteness issues of \theory{E}-unification using narrowing in the first-order approaches~\cite{Viola01,CLD05_TheFiniteVariantProperty,VariantNarrowing:EscobarMS09}. But these are not easily transferable to the nominal framework due to some limitations of the existing developments.    Most of the alternative approaches to first-order unification using narrowing techniques rely on restricting the equational theory \theory{E} to satisfy (some of the) following requirements:
\begin{enumerate}
    \item \theory{E} is regular, that is, for each $s\approx t \in \theory{E}$, we have $V(s) = V(t)$;
    \item  An  \theory{E}-unification algorithm exists;
    \item \theory{E} is finitary;
    \item \theory{R} is \theory{E}-confluent and \theory{E}-terminating;
    \item $\erew$ is \theory{E}-coherent.
\end{enumerate}
The current limitation for the nominal framework lies in requirements 2 and 3: the equational theories for which nominal unification has been developed are limited to \theory{C} and \theory{AC}. Besides,  the current developments identified that these equational theories interact with freshness and $\alpha$-renaming in a way that disrupts the finitary property of nominal \theory{C/AC}-unification (see Example~\ref{rmk:cunif-notfin}).
Alternative approaches that yield finitary theories (using, e.g., fixed-point equations to express solutions) are under development.

\section{Related Work}\label{sec:rel-work}
%\dani{From~\citet{NominalNarrowing16}}
Narrowing has traditionally been used to solve (first-order) equations in initial and free algebras modulo a set of equations. It is also used to integrate functional and logic programming~\citep{EquationalProgramming_Dershowitz,FoundationsOfLogicProgramming_Lloyd}.  Narrowing was originally introduced for theorem proving~\citep{Hullot80}, but nowadays it is used in type inference~\citep{TypelevelComputationUsingNarrowing_Sheard} and verification of cryptographic protocols~\citep{SymbolicReachability_Meseguer}, amongst other areas. Narrowing gives rise to a complete \theory{E}-unification procedure if \theory{E} is defined by a convergent rewrite system, but it is generally inefficient. Several strategies have been designed to make narrowing-based \theory{E}-unification procedures more efficient by reducing the search space (e.g., basic narrowing~\citep{Hullot80} and variant narrowing~\citep{VariantNarrowing:EscobarMS09}, the latter inspired by the notion of \theory{E}-variant~\citep{CLD05_TheFiniteVariantProperty}) and sufficient conditions for termination have been obtained~\citep{Hullot80,VariantNarrowing:EscobarMS09,ALPUENTE_TerminationofNarrowing}. %\st{In this paper we develop basic nominal narrowing strategies and associated termination conditions, and leave the study of other complete strategies for future work.}

Nominal unification is closely related to higher-order pattern unification~\citep{NomUnifHOPerspective_Levy} and and there has been prior work on equational extensions in this setting. Prehofer~\citep{prehofer1998solving} introduced higher-order narrowing along with variants--such as lazy narrowing, conditional narrowing, pattern narrowing--and explored their use as inference mechanisms in functional-logic programming (see also~\cite{Hanus-Prehofer}). Nominal extensions of logic and functional programming languages are already available (see, e.g., \citep{FreshML_PittsGabbay,NomLogicProgramming_CheneyUrban}), and nominal narrowing could play a similar role in these languages.

Initial efforts to address equational nominal unification included integrating the theories of \theory{A}ssociativity ($\approx_{\alpha,\theory{A}}$), \theory{C}ommutativity~($\approx_{\alpha,\theory{C}}$) and \theory{A}ssocia\-tiv\-ity-\theory{C}ommutativity ($\approx_{\alpha,\theory{AC}}$) with $\alpha$-equality~\citep{A-C-AC/tcs/Ayala-RinconSFN19}. Various algorithms for nominal unification modulo commutativity ($\theory{C}$-unification) and formalisations of their correctness in proof assistants %PVS and Coq
have been developed~\citep{OnNominalSyntax_AyalaFernandezSobrinho,FormalisingNomC-unif/mscs/Gabriel21}, as well as algorithms for nominal \theory{C}-matching~\citep{FormalisingNomC-unif/mscs/Gabriel21} and nominal \theory{AC}-matching~\citep{CICM:AyalaRinconFSKN23}.
%(matching is a special case of unification where the substitution $\sigma$ only applies in one side).
Further investigations into nominal unification include exploring a \textit{letrec} construct and atom variables~\citep{NomUnif/Schmidt-Schauss22}.
These developments reveal significant differences between first-order and nominal languages, highlighting the challenges of extending equational unification algorithms to languages with binders. For example, the theory of \theory{C}-unification has nullary  type if $\alpha$-equivalence is considered~\citep{FormalisingNomC-unif/mscs/Gabriel21} contrasting with the finitary type of first-order \theory{C}-unification~\citep{Baader98}, and a direct extension of the  Stickel-Fages first-order \theory{AC}-unification algorithm was shown to introduce cyclicity in solutions~\citep{CICM:AyalaRinconFSKN23},
%produced by translations of nominal unification problems to Diophantine systems \cite{CICM:AyalaRinconFSKN23},
differing from the original first-order approach.

The techniques presented here for the definition of nominal narrowing modulo $\alpha$ and $\theory{E}$ could serve as a basis to extend Prehofer's results to deal with higher-order pattern-unification modulo equational theories $\theory{T}$ that cannot be presented by convergent rewrite rules, but could be represented by a convergent $\theory{R}$ modulo $\theory{E}$.
Moreover, nominal $\theory{R,E}$-narrowing could serve as a basis for extensions of first-order rewriting-based languages, such as Maude. The use of narrowing modulo equational axioms  makes Maude an excellent platform to develop applications such as security protocol analysis, and incorporating $\alpha$ will facilitate applications that require manipulating  syntax with binding.

\section{Conclusion and Future work}
\label{sec:future-work}

In this work, we present an approach to solving equational nominal unification problems via narrowing extended with axioms. We proposed definitions for nominal $\theory{E}$-rewriting and $\theory{E}$-narrowing and proved some properties relating them, obtaining  the \theory{E}-Lifting Theorem, when \theory{R}  is an \theory{E}-convergent NRS, $\to_\theory{R,E}$ is \theory{E}-coherent and a complete algorithm for nominal \theory{E}-unification exists.
We illustrated the technique with an example of symbolic differentiation.
% and security protocol analysis.
In future work we will explore applications in rewriting-based programming languages and theorem provers, as well as applications in the symbolic analysis of security protocols. Several tools, such as Maude~\citep{DBLP:journals/jlap/LopezRuedaES23}, Deepsec~\citep{DBLP:journals/theoretics/ChevalKR24}, Tamarin~\citep{DBLP:journals/ieeesp/BasinCDS22}, ProVerif~\citep{ProVerif/csfw/ChevalR23}, etc, use equational reasoning.

The current working hypothesis to constrain approaches to theorem proving and logic programming is that it is not necessary to compute finite complete sets of unifiers: it is sufficient to decide the solvability of the \theory{E}-unification problems.  Viola~\cite{Viola01} proposed an alternative narrowing strategy by introducing rule extensions.
Other alternative methods, which we aim to work on in the future, are to use Variant Narrowing~\citep{VariantNarrowing:EscobarMS09,Escobar12_FoldingVariantNarrowing} or Needed Narrowing~\citep{10.1145/347476.347484}.
% and consider the use of efficient narrowing strategies such as variant narrowing~\citep{VariantNarrowing:EscobarMS09} and needed narrowing~\citep{10.1145/347476.347484}.
% \citet{VariantNarrowing:EscobarMS09} proposed a new notion of
Variant Narrowing works
with rules \theory{R} modulo axioms \theory{E} and is complete for any \theory{R{\cup}E} satisfying some key conditions, and avoids many unnecessary narrowing sequences that would result from complete narrowing. Also, if \theory{R} satisfies the finite variant property modulo \theory{E}~\citep{CLD05_TheFiniteVariantProperty}, then it can be refined into a terminating and complete narrowing algorithm.

The finite variant property enables the reduction of problems modulo an equational theory \theory{E} to problems modulo a subtheory $\theory{E'} \subseteq \theory{E}$~\citep{CLD05_TheFiniteVariantProperty}.
It can be essential for
constraint-solving problems, such as intruder derivability constraints~\citep{ComonLundh_IntruderDeductions} or disunification problems, which explains our desire to study and work with these approaches in the future.

\bibliographystyle{elsarticle-harv}
\bibliography{mybibliography}

\end{document}